\documentclass[lettersize,journal]{IEEEtran}

\usepackage{cite}
\usepackage{amsmath,amssymb,amsfonts}
\usepackage{algorithmic}
\usepackage{graphicx,color}
\usepackage{blindtext}
\usepackage{hyperref}
\usepackage{amsthm}
\usepackage{caption}
\usepackage{subcaption}
\usepackage{xcolor}
\usepackage{tikz}
\usepackage{pgfplots}
\usetikzlibrary{positioning}
\usetikzlibrary{fit}
\usepackage{mathtools}
\usepackage{diagbox}

\newcommand{\dist}{\mathtt{dist}}
\newcommand{\dd}{\mathtt{D}}
\newcommand{\safe}{\mathtt{safe}}
\newcommand{\diag}{\mathtt{diag}}
\newcommand{\dest}{\mathtt{dest}}
\newcommand{\EO}{\mathtt{EO}}
\newcommand{\Loss}{\mathtt{Loss}}

\allowdisplaybreaks

\newtheorem{theorem}{Theorem}
\newtheorem{prop}{Proposition}
\usepackage{textcomp}
\begin{document}

\title{Data-Driven Robust Optimization for Energy-Aware Safe Motion Planning of Electric Vehicles}

\author{Simran Kumari, Ashish R. Hota and Siddhartha Mukhopadhyay
\thanks{The authors are with the Department of Electrical Engineering, IIT Kharagpur, West Bengal, India. Email: simranjnr@kgpian.iitkgp.ac.in, \{ahota,smukh\}@ee.iitkgp.ac.in. Simran Kumari is supported by Ministry of Education (MoE), Govt. of India under Prime Minister Research Fellowship (PMRF) scheme.}
}

\maketitle

\begin{abstract}
In this paper, we simultaneously address the problems of energy optimal and safe motion planning of electric vehicles (EVs) in a data-driven robust optimization framework. Safe maneuvers, especially in urban traffic, are characterized by frequent lateral motions, such as lane changes, overtakes and turning along curved roads. Motivated by our previous work which shows a $3-10 \%$ increase in energy consumption due to lateral motion when an electric vehicle changes its lane once every kilometer while following standard drive cycles, we incorporate vehicle lateral dynamics in the modeling and control synthesis, which is in contrast with most prior works. In the context of safety, we leverage past data of obstacle motion to construct a future occupancy set with probabilistic guarantees, and formulate robust collision avoidance constraints with respect to such an occupancy set using convex programming duality. Consequently, we formulate a finite-horizon optimal control problem subject to robust collision avoidance constraints while penalizing resulting energy consumption, and solve it in a receding horizon fashion. Finally, we show the effectiveness of the proposed approach in reducing energy consumption and collision avoidance via numerical simulations involving curved roads and multiple obstacles. A detailed analysis of energy consumption along different components of EV motion highlights appreciable improvement under the proposed approach.
\end{abstract}


\section{Introduction}
Electric vehicles (EVs), being a cleaner alternative to vehicles with internal combustion engines, have seen significant growth in the past few years. Advanced Driver Assist Systems (ADAS) for EVs need to fulfil two primary objectives. First, it is essential for ADAS to ensure safe collision-free operation of the vehicle in the uncertain environment of roads with nearby vehicles represented as dynamic obstacles. The second objective is to consider energy efficient driving in order to be able to extend the driving range in the face of battery storage limitations. In this work, we propose a data-driven robust receding-horizon optimization framework for energy efficient and safe motion planning of EVs.

In the context of safety, there exists a substantial body of literature on path planning and collision avoidance for (autonomous) vehicles \cite{lefevre2014survey,gonzalez2015review,dixit2019trajectory}. Receding-horizon or model predictive control (MPC) based approaches are being used widely for ADAS functionalities such as lane change, overtaking and collision avoidance, due to its real-time planning ability while adhering to constraints imposed by vehicle dynamics and obstacles \cite{dixit2019trajectory,zhang2019trajectory,zhang2020optimization}. Specifically, \cite{zhang2020optimization} formulated collision avoidance constraints for convex polytopic representation of obstacles. However, these works assumed obstacles to be static and ignored the dynamic evolution of their positions over time. 

Several recent papers have tackled the problem of safe motion planning in presence of dynamic obstacles using data-driven or stochastic MPC techniques. The work \cite{schildbach2015scenario} utilized scenario MPC for trajectory planning during lane change operation initiated by the driver. Similarly, \cite{brudigam2021stochastic} proposed a stochastic MPC formulation with collision avoidance constraints by approximating the motion of surrounding vehicles by linear dynamics. The authors in \cite{qie2022improved} combined physics-based and data-driven intention-based models for future trajectory prediction of obstacles, and formulated a tube-based MPC to track a given reference trajectory generated by considering a predicted obstacle trajectory. Similarly, \cite{benciolini2023non} presented a trajectory planning framework which estimated intentions of dynamic obstacles. The authors in \cite{zhou2022interaction} predicted multi-modal probabilistic occupancy of surrounding vehicles for a given maneuver set, and generated target trajectory while maintaining safe distance from the predicted occupancy set. Authors in \cite{hakobyan2020wasserstein,navsalkar2023data} formulated distributionally robust MPC for safe navigation of mobile robots in uncertain environments; however the resulting optimization problems scaled with the number of data points leading to higher computational burden. However, the above works did not examine resulting energy consumption due to safe maneuvers, which is critical in the context of EVs.

Several past works have also explored the problem of energy-efficient driving in EVs, and proposed eco-driving \cite{lee2020model,shen2020minimum,dib2014optimal} and eco-routing \cite{yi2018optimal} solutions. These approaches generally obtained optimized EV speed profile over distance within $2$ km considering longitudinal dynamics, resulting in overestimated driving range since road curvature and lateral motion were neglected; further, safety issues brought on by nearby traffic participants were ignored. There are a few works that simultaneously incorporate the effect of energy consumption and safety in the trajectory or motion planning problem. For instance, \cite{sajadi2019nonlinear,kamal2010ecological} proposed the ecological driver assistance system (EDAS), and mostly examined the Adaptive Cruise Control (ACC) mode. EDAS provided short-range optimized speed profile (generally over a window of $10$ seconds), while maintaining a safe distance (about $15$ or $20$ m) with respect to a preceding vehicle. Moreover, EDAS considered only the effects of longitudinal motion on energy consumption. The authors in \cite{hausler2015energy} tackled a similar problem for generic mobile robots; however, they modeled the robot dynamics by the kinematic bicycle model which ignores forces acting at the wheels, and it is not feasible to relate battery energy consumption of an EV with only kinematic variables.

However, real world safe driving, especially urban driving, includes far more complex scenarios, consisting of discrete events of overtaking, lane changes and double lane changes during maneuver, and is therefore dominated by repeated lateral motion. In fact, past works \cite{knoop2012quantifying,yang2018effect,asaithambi2017overtaking} have shown that a vehicle typically changes its lane once every $0.5-1$ km. In our previous work \cite{latdyn}, we showed that there exists $3-10 \%$ increase in energy consumption due to lateral motion, leading to a significant difference in the driving range, particularly a reduction of $15$, $23$, and $36$ km for FTP-$75$, NEDC, and HWFET drive cycles, respectively, for a $54.28$ kWh battery given that the vehicle changes its lane once every $1$ km. In addition to lane change, curved roads also result in EV lateral motion and additional energy consumption compared to straight roads. The magnitude of the driving range reduction grow noticeably with a reduction in the radius of curvature of the road \cite{ding2018optimal,bifulco2021combined} as illustrated in Section III of the paper.

In light of the above discussion, we propose a data-driven robust optimization framework for energy-aware safe motion planning of an EV that incorporates the impact of lateral motion, and highlight its effectiveness via realistic simulations. Our contributions are summarized below.
\begin{itemize}
    \item We consider a dynamic model of the EV that captures both longitudinal and lateral motion as well as the dynamics of stored energy level in the battery, which is a departure from most past works that only consider vehicle longitudinal dynamics. In contrast with longitudinal motion models, the proposed model captures realistic driving scenarios, which includes lane changes, overtaking, and curved roads as well as provides more accurate energy consumption quantification.
    \item Instead of adopting a model-based approach to predict the future trajectory of the obstacles, we adopt a data-driven approach where past data on the motion of surrounding obstacles is used to construct a polyhedral occupancy set for each obstacle over a finite prediction horizon. These sets are constructed using Principal Component Analysis (PCA), and are guaranteed to contain the future positions of the obstacles with high probability following prior works \cite{cheramin2021data,shang2019data}. This approach easily integrates measurements obtained from on-vehicle sensors such as cameras, Lidars, and radars.
    \item We then derive robust collision avoidance constraints which guarantee that the EV will not intersect with the future occupancy sets of the obstacles over a finite prediction horizon. A multi-stage optimization problem is formulated which minimizes control efforts, distance from a specified point, and energy consumption from the battery, while satisfying robust collision avoidance constraints, and bounds on states and control inputs. 
\end{itemize}

We evaluate the effectiveness of the proposed approach via simulations for different scenarios, including motion on a straight road, curved road, and with multiple obstacles. We compare the proposed approach with an energy-unaware approach (where the penalty on energy consumption is removed from the cost function) and provide a detailed analysis of components of energy consumption due to longitudinal and lateral motion. Our results show that under the proposed approach, the EV is able to maneuver ahead while avoiding obstacles, and consumes significantly less amount of energy compared to the energy-unaware scheme with minimal change in travel time. We further highlight the significance of energy consumption due to lateral motion during overtaking and on curved roads, and the efficacy of the proposed scheme in improving the driving range. 


\section{Energy-Aware Safe Motion Planning Framework}\label{prblm_des}
We consider the problem of energy-aware safe motion planning as an optimal control problem with robust collision avoidance constraints. The components of the stated problem, including cost function, EV dynamics, and collision avoidance constraints, formulated leveraging past data, are described in this section.  

\subsection{Plant Model} 

We consider a bicycle model \cite{kong2015kinematic} of a rear-wheel-driven EV with states given by coordinate of the center of gravity (COG) along longitudinal axis of EV ($s_x$) in the body frame, lateral deviation ($e_y$) and relative yaw angle ($e_{\psi}$) with respect to center-line of road, position of COG of the EV ($p_x, p_y$) in the inertial frame, velocity along longitudinal and lateral axes of the EV ($v_x, v_y$) in the body frame, rotation of the vehicle frame about the $z$-axis ($\psi$) in the inertial frame, yaw rate ($r$) in the body frame, and state-of-energy of the battery ($\gamma$). There are three inputs being applied to the EV: acceleration along longitudinal axis ($a$), braking deceleration along longitudinal axis ($d$), and steering angle ($\delta$). The continuous-time dynamics of EV with state vector $x :=[s_x, e_y, e_{\psi}, v_x, v_y, r, \gamma, p_x, p_y, \psi]^\intercal \in \mathbb{R}^{10}$ and input vector $u :=[a, \delta, d]^\intercal \in \mathbb{R}^{3} $ are given by
\begin{subequations}\label{eq:plant_dynamics}
\begin{align}
\dot{s}_x&=v_x,\label{1}\\
\dot{e}_y&=v_x+e_{\psi}v_y,\\
\dot{e}_{\psi}&=v_y-v_x\rho,\\
\dot{v_x}&=\frac{F_{\mathrm{x}}}{m}-F_{\mathrm{a}}-F_{\mathrm{r}}-\left(\frac{F_{F,y}\sin\delta-mv_y r}{m} \right)\\ \nonumber
&=a+d-F_{\mathrm{a}}-F_{\mathrm{r}}-\left(\frac{F_{F,y}\sin\delta-mv_y r}{m} \right),\\
\dot{v_y}&=\frac{(F_{F,y}\cos\delta+F_{R,y}-mv_xr)}{m},\\
\dot{r}&=\frac{(F_{F,y}l_F\cos\delta-F_{R,y}l_R)}{I_{\mathrm{z}}},\\
\dot{\gamma}&= \frac{-\eta_\mathrm{b} P_\mathrm{b}}{E_\mathrm{b}},\\
\dot{p}_x&= v_x\cos\psi-v_y\sin\psi,\\
\dot{p}_y&= v_x\sin\psi+v_y\cos\psi,\\
\dot{\psi}&=r. \label{1_end}
\end{align}
\end{subequations}

\begin{table}[t!]
\centering
\caption{Nomenclature}
   \begin{tabular}{|p{6cm}|p{1.5cm}|}
 \hline
 \textbf{Parameters} & \textbf{Symbols} \\
 \hline \hline
  Mass & $m$ \\
  Yaw polar inertia & $I_{\mathrm{z}}$\\
  Coefficient of drag & $C_\mathrm{d}$\\
   Air density & $\rho_\mathrm{a}$\\
    Frontal area of EV & $A_\mathrm{f}$\\
    Rolling resistance coefficient & $C_\mathrm{r}$\\
   Longitudinal distance from COG to front axle & $ l_F $ \\
     Longitudinal distance from COG to rear axle & $ l_R $ \\
        slip angle at front tire  & $ \alpha_F $ \\
         Slip angle at rear tire & $ \alpha_R $ \\
              Front tire corner stiffness  & $ C_{F} $ \\
              Rear tire corner stiffness  & $ C_{R} $\\ 
              Lateral force at front wheel  & $F_{F,y}$\\
              Lateral force at rear wheel & $F_{R,y}$\\
              Longitudinal force acting at COG of vehicle & $F_\mathrm{x}$\\
              Curvature of road & $\rho$\\
Differential gear ratio & $N_\mathrm{d}$\\
Wheel radius & $r_\mathrm{w}$\\
Drivetrain efficiency & $\eta_\mathrm{i}$\\
Battery charging/discharging efficiency & $\eta_\mathrm{b}$\\
Energy storage capacity of the battery & $E_\mathrm{b}$\\
Battery power & $P_\mathrm{b}$\\
 Motor Torque & $\tau_\mathrm{m}$\\
 Motor speed & $\omega_\mathrm{m}$\\
\hline    
    \end{tabular}
    \label{tab:my_label}
\end{table}

The notation used in the above equations for different EV parameters are summarized in Table~\ref{tab:my_label}, and a schematic of the model is presented in Figure \ref{model}. In particular, the force terms $F_{F,y}$, $F_{R,y}$, $F_\mathrm{a}$ and $F_\mathrm{r}$ are given by
\begin{subequations}
\begin{align}
F_{F,y}&=-2C_{F}\alpha_{F},\quad F_{R,y}=-2C_{R}\alpha_{R},\\
\alpha_F&= \arctan\left(\frac{v_y+l_Fr}{v_x}\right)-\delta,\\
\alpha_R&= \arctan\left(\frac{v_y-l_Rr}{v_x}\right),\\
F_\mathrm{a}&=\frac{0.5\rho_\mathrm{a} C_\mathrm{d} A_\mathrm{f} v_{x}^2}{m},\\
F_\mathrm{r}&=C_\mathrm{r}g.
\end{align}
\end{subequations}

\begin{figure}[t]
\centerline{\includegraphics[width=0.45\textwidth,keepaspectratio]{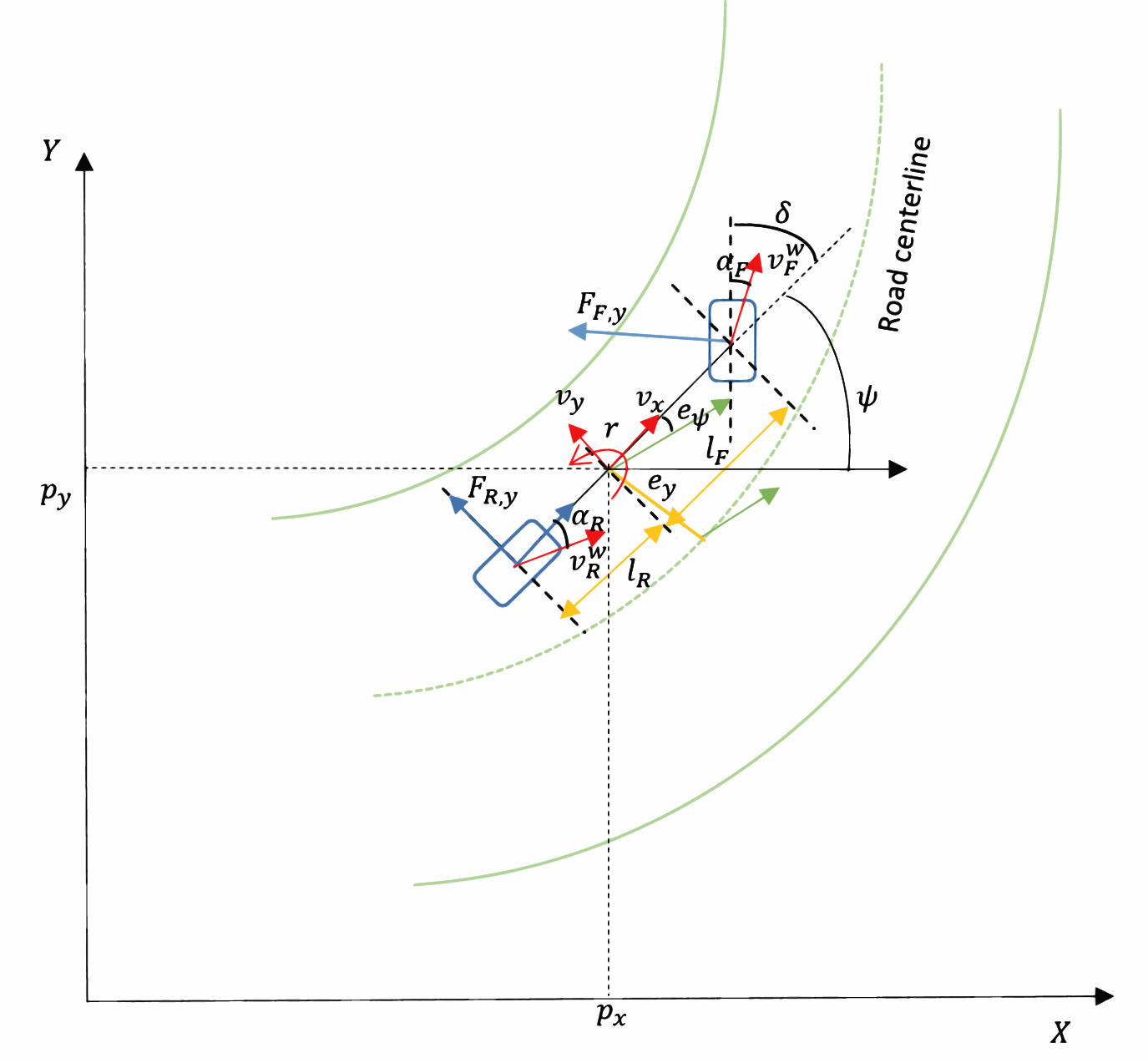}}
\caption{\footnotesize Schematic of different kinematic and dynamic variables associated with the bicycle model.}
\label{model}
\end{figure}
For a highly efficient inverter, the battery power $P_\mathrm{b}$ (assumed to be positive during discharging and negative during charging/regeneration) can be approximated in terms of motor input power and motor power loss, both of which depend on the motor torque $\tau_\mathrm{m}$ and motor speed $\omega_\mathrm{m}$. Specifically, 
\begin{subequations}\label{eq:battery_power}
\begin{align}
    \tau_\mathrm{m}&=\frac{mr_\mathrm{w}}{\eta_\mathrm{i}N_\mathrm{d}}a,\quad \omega_\mathrm{m}=\frac{N_\mathrm{d}}{r_\mathrm{w}} v_x,
    \\ P_\mathrm{b}&=\tau_\mathrm{m}\omega_\mathrm{m}+P_{\Loss},
    \\ P_{\Loss}&=k_\mathrm{c}\tau_\mathrm{m}^2+k_\mathrm{i}\omega_\mathrm{m}+k_\mathrm{w}\omega_\mathrm{m}^3+k_\mathrm{f}\omega_\mathrm{m}^2+k_0, 
\end{align}
\end{subequations}
where $P_{\Loss}$ represents power loss in the motor assumed to be a function of $\tau_\mathrm{m}$ and $\omega_\mathrm{m}$ with coefficient of copper loss $k_\mathrm{c}$, iron loss coefficient $k_\mathrm{i}$, the windage loss coefficient $k_\mathrm{w}$, the friction loss coefficient $k_\mathrm{f}$, and $k_0$ being a constant component \cite{larminie2012electric}. Consequently, we approximate battery power as a cubic polynomial of the motor speed and motor torque given by 
\begin{align}
P_b(\omega_\mathrm{m},\tau_\mathrm{m}) & = c_1+c_2\omega_\mathrm{m}+c_3\tau_\mathrm{m}+c_4\omega_\mathrm{m}^2+c_5\omega_\mathrm{m}\tau_\mathrm{m}\nonumber
\\ & \quad +c_6\tau_\mathrm{m}^2+c_7\omega_\mathrm{m}^3.\label{eq:battery_power_approx}
\end{align}
The coefficients are obtained using least squares regression. For simplicity, we assume that efficiency during charging and discharging operation of the battery are identical.

Note that we have included position and orientation in both inertial and body coordinate systems as state variables in order to capture kinematic as well as dynamic behavior of EV. Previewed road curvature $\rho$ is estimated as a function of $s_x$ given by
\begin{equation}
    \rho(s_x) := k_1s_x^2+k_2s_x+k_3, \label{curv}
\end{equation}
with the coefficients obtained using least square regression.

The continuous-time model \eqref{eq:plant_dynamics} can be represented in compact form as $\dot{x} =f(x,u)$.\footnote{We assume that all states are measurable, and sensors provide accurate measurement.} The discrete-time EV model is obtained after Euler discretization\footnote{Other suitable discretization techniques may also be applied.} of the continuous-time model with sampling time $T_\mathrm{s}$, and is given by
\begin{align}
x_{t+1}&=x_t+T_\mathrm{s}f(x_t,u_t).\label{dyyn}
\end{align}

\subsection{Nominal Constraints}
The EV dynamic variables need to satisfy the following constraints on the inputs and states:
\begin{subequations}
\begin{align}
    a_{\min}&\leq a \leq a_{\max}, \label{12}\\
    d_{\min}&\leq d \leq d_{\max}, \label{12a} \\
    \delta_{\min}&\leq \delta \leq \delta_{\max}, \label{13}\\
    v_{x,\min}&\leq v_x\leq v_{x,\max}, \label{sb1} \\
    e_{y,\min}&\leq e_y\leq e_{y,\max}, \label{sb2} \\
    \gamma_{\min}& \leq \gamma \leq \gamma_{\max} \label{sb3}.
\end{align}
\end{subequations}
These include actuation constraints \eqref{12}, \eqref{12a}, \eqref{13}, longitudinally forward movement of EV \eqref{sb1}, movement within road boundary \eqref{sb2}, and safe operating region of battery \eqref{sb3}. For a specific curved road, high speed maneuver requires large lateral acceleration and yaw rate. If lateral tyre forces are saturated due to non-linearity and limitations of tyre vertical loads, sufficient lateral acceleration or yaw rate may not be generated. This leads to loss of vehicle stability such as oversteer and understeer conditions \cite{zhou2021vehicle}. Therefore, the value of the maximum speed limit $v_{x,\max}$ is chosen based on the radius of curvature of the road according to \cite{amata2008abrupt}, and the constraint \eqref{sb1} plays a critical role in ensuring lateral stability of the vehicle. Note that constraints on the rate of the change of the above variables may also be included with minimal increase the complexity of the optimization problem.

\begin{figure}[t]
\centerline{\includegraphics[width=0.4\textwidth,keepaspectratio]{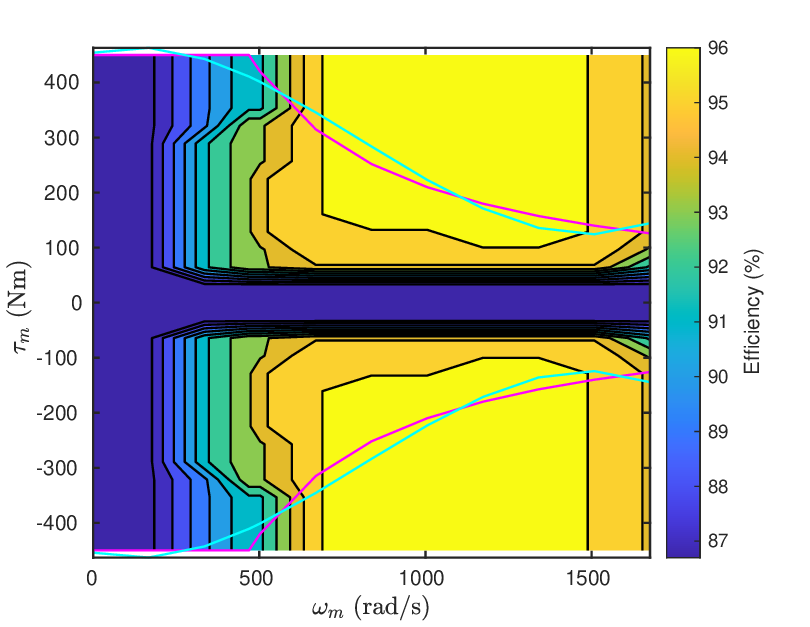}}
\caption{\footnotesize Motor efficiency plot. Actual and estimated motor torque-speed curve is plotted in red and cyan respectively. {The motor torque is required to remain bounded between the two cyan curves according to \eqref{trq} and \eqref{torque_limit}.}}
\label{fig:motor}
\end{figure}

In addition, motivated by typical motor torque-speed efficiency curve \cite{motor}, we impose the following constraints on the motor torque ($\tau_\mathrm{m} = \frac{m r_\mathrm{w}}{\eta_\mathrm{i}N_\mathrm{d}} a$):
\begin{subequations}
\begin{align}
& -\tau_{\mathrm{m},\max}(v_x) \leq  \frac{m r_\mathrm{w}}{\eta_\mathrm{i}N_\mathrm{d}} a\leq \tau_{\mathrm{m},\max}(v_x), \quad \text{where} \label{trq}\\
& \tau_{\mathrm{m},\max}(v_x)=\xi_1v_x^3+\xi_2v_x^2+\xi_3v_x+\xi_4. \label{torque_limit}
\end{align}
\end{subequations}
In particular, the torque-speed efficiency curve is approximated by a cubic polynomial of motor speed $\omega_\mathrm{m}$ (and thus, longitudinal speed $v_x$); see Figure \ref{fig:motor} for an illustration. The coefficients are determined using least squares regression. 

Note that we assume a constant drivetrain efficiency, denoted as $\eta_\mathrm{i}$, for the components involved in transmitting power from the propulsion source to the wheels. This includes elements such as the transmission, driveshaft, differential, axles, and joints, but does not include the motor. We incorporate a dynamic efficiency map between motor input and motor output, as illustrated in Figure \ref{fig:motor}. The powertrain consists of both the drivetrain and the motor, and its efficiency is dynamic because the motor efficiency is dynamic.


\subsection{Data-Driven Robust Collision Avoidance}

We now formulate the collision avoidance constraints. We assume that the ego EV is a point object represented by its position $p_t:= [p_{x,t},p_{y,t}]^\intercal\in\mathbb{R}^2$ at a discrete time instant $t$. Let there be $M$ number of obstacles, and $\mathbb{O}^{(m)}_t\subset\mathbb{R}^{2}$ be the polytope enclosing the space occupied by $m$-th obstacle at time instant $t$. Then, to guarantee collision free movement of the ego EV over a finite horizon of length $N$, we require the following condition to hold:
\begin{align}
 \{p_k\} \cap\mathbb{O}^{(m)}_k=\emptyset, \quad \forall m\in [M],\forall k \in \{t,\ldots, t+N\},
\end{align} 
where $[M] := \{1,2,\ldots M\}$ for brevity of notation. 
The safety constraint can be formulated in terms of a distance function:
\begin{equation}
\dist(p_k,\mathbb{O}^{(m)}_k)\!=\!\min_{y} \{||y||_2:(p_k+y) \in \! \mathbb{O}^{(m)}_k\} \! > \! \dd_{\safe},\label{eq:safety_constraint_def}
\end{equation}
where $\dd_{\safe} > 0$ represents a safety margin on the distance between the ego EV and the obstacle, and $y$ is a vector in $\mathbb{R}^2$ space. We omit the superscript $(m)$ for better readability since the discussion is applicable for any of the $M$ obstacles. For simplicity, we assume a constant value of $\dd_{\safe}$ for each obstacle.

At a given time $t$, the polytope enclosing the space occupied by an obstacle is given by
\begin{align}
\mathbb{O}_t&=\{y\in \mathbb{R}^2:Ay\leq b\}.
\end{align}
Since we need to impose collision avoidance constraints over a prediction horizon of finite length $N$, we need to make suitable assumptions on how the occupancy set of such an obstacle evolves over time. In particular, we assume that the orientation of this obstacle in the inertial frame remains unchanged, but this set $\mathbb{O}_t$ may be translated by a random vector in $\mathbb{R}^2$. Specifically, the occupied space at discrete time instant $t+k$, denoted $\mathbb{O}_{t+k}$, is represented as
\begin{align}\label{dyn_obs}
\mathbb{O}_{t+k}= \mathbb{O}_t \oplus \Omega_{t}^{k},
\end{align}
where, $\oplus$ denotes Minkowski sum of sets, and $\Omega^{k}_{t}$ is the support of random vector that denotes $k$th step translation of the obstacle at time $t$. The support set should be large enough to capture the true realization of uncertainty with high confidence irrespective of the underlying distribution and should exclude scenarios that violate hard constraints such as road space constraints. In this regard, we present a data-driven approach to construct the set $\Omega^{k}_{t}$ from past data on obstacle position followed by robust collision avoidance constraints for this set. 

\subsubsection{Data-Driven Uncertainty Set Generation}\label{framework}

At each time, we observe the current position of the obstacle and obtain the change in its position from its position $k$ time steps earlier, for $k \in [N]$. Formally, let $\omega^j_{j-k} \in \mathbb{R}^2$ denote the difference between the position of the obstacle at current time $j$ and the position of the obstacle $k$ steps prior to the current time. Let us assume that translation data is being sampled at time interval $T_f$. For each $k \in [N]$, we collect $N_s$ number of such samples or scenarios denoted by $W^k_t=\{\omega^t_{t-kf^{'}},\ldots, \omega^{t-N_sT_f}_{t-N_sT_f-kf^{'}} \}\subseteq \mathbb{R}^2$ from recent past $N_s$ time points corresponding to current time $t$  where $f^{'}=\frac{T_f}{T_s}$. We now construct a polyhedral uncertainty set in $\mathbb{R}^2$, denoted by $\Omega(W^k_t)$, using Principal Component Analysis (PCA) following \cite{cheramin2021data,shang2019data}. We omit the superscript $k$ in the following paragraph for better readability. 

We first compute the sample mean of $W_t$ as $\bar{\omega}:=\frac{1}{N_s}\sum_{j=0}^{N_s}\omega_{t-jT_s-f^{'}}^{t-jT_s}$. We then subtract this mean from each scenario to obtain the ``centered" data $\Delta\omega_{t-jT_s-f^{'}}^{t-jT_s}=\omega_{t-jT_s-f^{'}}^{t-jT_s}-\bar{\omega}$ for every sample in $W_t$. We construct the data matrix $\textbf{Z}_0 \in \mathbb{R}^{N_s\times 2}$, where $j$-th row represents $(\Delta\omega_{t-jT_s-f^{'}}^{t-jT_s})^\intercal\in\mathbb{R}^{1\times2}$, and compute the sample covariance matrix
\begin{align*}
\textbf{C}&=\frac{1}{N_s-1}\textbf{Z}^\intercal_0\textbf{Z}_0.
\end{align*}
With $ \textbf{Z}^\intercal_0=U\Sigma V^\intercal $, we perform eigen-decomposition on \textbf{C} as
\begin{align*}
\textbf{C}&=U\left(\frac{\Sigma\Sigma^\intercal}{N_s-1}\right)U^\intercal=U\Lambda U^\intercal,
\end{align*}
where $U=[d_1,d_2]\in \mathbb{R}^{2\times2}$, $V\in \mathbb{R}^{N_s\times N_s}$ and $\Sigma\in \mathbb{R}^{2\times N_s} $. The columns of $U$, $d_1$ and $d_2$, are principal directions of the centered data. Following \cite{cheramin2021data,shang2019data}, we obtain a polyhedral uncertainty set approximation given by
\begin{align*}
& \hat{\Omega}_t(W_t)=\biggl\{ y:y=\bar{\omega}+\sum_{i=1}^{2}\big(\alpha_i \frac{\bar{\omega}_i}{||d_i||}d_i +(1-\alpha_i)\frac{\underline{\omega}_i}{||d_i||}d_i\big),\\
& \quad \qquad \qquad 0\leq \alpha_i \leq 1, \forall i\in \{1,2\} \biggl\}, \quad \text{where}\\
& \bar{\omega}_i=\max_{1\leq j\leq N_s}\left\lbrace \frac{\omega_{j0}^\intercal d_i}{||d_i||}\right \rbrace\in \mathbb{R},\quad \underline{\omega}_i=\min_{1\leq j \leq N_s}\left\lbrace \frac{\omega_{j0}^\intercal d_i}{||d_i||} \right \rbrace\in \mathbb{R}.
\end{align*}

We now present the following result from \cite{cheramin2021data} which states probabilistic guarantees on $\hat{\Omega}_t(W_t)$ defined above.
\begin{theorem}[\cite{cheramin2021data},Theorem 2]\label{thrm1}
If $N_{s}\geq \lceil \frac{1}{\epsilon}\frac{e}{e-1}(3+\ln\frac{1}{\beta})\rceil$, then we have $1-\beta$ confidence that any realisation $\omega$ belongs to the uncertainty set $\hat{\Omega}_t(W_t)$ with the probability of at least $1-\epsilon$, i.e.
\begin{align*}
  \mathbb{P}_{W_t}\{\mathbb{P}_{\omega}\{\omega\in \hat{\Omega}_t(W_t)\}\geq1-\epsilon\}\geq 1-\beta,  
\end{align*}
where $\epsilon\in(0,1), \beta \in (0,1)$.
\end{theorem} 

Thus, for each $k \in [N]$, we construct a polyhedral uncertainty set $\hat{\Omega}_t^k(W_t^k)$ from past data that contains $k$th step change in the position of the obstacle with high probability. The set $\hat{\Omega}_t^k(W_t^k)$ is oriented along the principal directions of the centered data and encloses all points in the dataset $W^k_t$. It can be equivalently expressed in the form 
\begin{align}
\hat{\Omega}^k_t(W^k_t)=\{\omega\in \mathbb{R}^2:\hat{G}^k_t \omega\leq \hat{h}^k_t\}. 
\end{align}

In addition, we know that the future obstacle position should remain confined to the road boundary. We inner-approximate the road area around an obstacle by a polyhedron $\Delta$. Then, ${\Omega}^k_t(W^k_t)$ should adhere to the condition
\begin{align}
    {\Omega}^k_t(W^k_t)\subseteq \Delta \oplus (-\mathbb{O}_t).
\end{align}
Consequently, we define the uncertainty set ${\Omega}^k_t(W^k_t)$ as
\begin{align}
    {\Omega}^k_t(W^k_t)= (\Delta \oplus (-\mathbb{O}_t) )\cap \hat{\Omega}^k_t(W^k_t),
\end{align}
which can be equivalently stated in the form
\begin{align}
{\Omega}^k_t(W^k_t)=\{\omega\in \mathbb{R}^2:{G}^k_t \omega\leq {h}^k_t\}. \label{eq:uncertainty_set}  
\end{align}

\subsubsection{Robust Collision Avoidance Constraints}

In this subsection, we present the robust collision avoidance constraints for the ego EV. We treat the data-driven uncertainty set $\Omega(W^k_t)$ defined in \eqref{eq:uncertainty_set} to be the support of the random variable that represents $k$th step translation of the obstacle in \eqref{dyn_obs}. We now reformulate the safety constraint \eqref{eq:safety_constraint_def}.

\begin{prop}\label{dynamic}
Assume that at time $t+k$, the obstacle occupancy set $\mathbb{O}_{t+k} := \mathbb{O}_t \oplus {\Omega}^k_t(W^k_t)$, the EV position is given by $p_{t+k}$, and let $\dd_{\safe} > 0$ be a desired safety margin between the controlled object and the obstacle. Then, we have
\begin{align*}
& \dist(p_{t+k},\mathbb{O}_{t+k})> \dd_{\safe}\\
\Leftrightarrow & \dist(p_{t+k},\mathbb{O}_t\oplus\{\omega^{k}\})> \dd_{\safe}, \quad \forall \omega^{k} \in \Omega^k_t(W^k_t) \\
\Leftrightarrow & \exists\lambda_k \geq 0,\mu_{k} \geq 0:\dd_{\safe}-(A p_{t+k} -b)^\intercal \lambda_k+\mu_{k}^\intercal h^k_t<0,\\
 & \qquad \qquad \qquad \quad ||A^\intercal\lambda_k||_2\leq 1,A^\intercal\lambda_k=G^{k^\intercal}_t\mu_{k},
\end{align*}
where $\Leftrightarrow$ indicates ``if and only if" condition.
\end{prop}

The proof is presented in Appendix \ref{appendix1}. The constraints given in the above proposition guarantee that at time $t+k+1$, the ego EV is at least $\dd_{\safe}$ distance away from the occupancy region of the obstacle for all possible change in its position from the data-driven uncertainty set $\Omega^k_t(W^k_t)$. Since the true displacement of the obstacle is expected to lie within the uncertainty set with high probability, it follows from Theorem \ref{thrm1} that the distance between the ego EV and the obstacle is guaranteed to be larger than $\dd_{\safe}$ with high probability.

A static obstacle remains at its initial position for all time. In that case, uncertainty set $\Omega^k_t(W^k_t) := \{0\}$. Consequently, constraints in Proposition~\ref{dynamic} reduce to
\begin{align*}
  & \dist(p_{t+k},\mathbb{O}_t)> \dd_{\safe}\\
\Leftrightarrow & \exists \lambda_k \geq 0:(A p_{t+k}-b)^\intercal \lambda_k>\dd_{\safe},
||A^\intercal\lambda_k||_2\leq1, 
\end{align*}
which coincides with the constraints developed in \cite{zhang2020optimization} for static obstacles. 


\subsection{Constrained Optimal Control Problem Formulation} 

We now formulate the constrained optimal control problem for energy-aware and safe motion planning of the ego EV. The goal of the ego EV is to minimize travel time and maintain the lane while ensuring safety. The set points ($v_{x,\max}, e_{y,\dest})$ are assumed to be provided by the route level planner module or the driver. The $e_{y,\dest}$ component captures the geometry of the road. Let the state $x_t$ of ego EV be known at the current time $t$. The notation $u_{t:(t+N-1)}:= \{u_t,u_{t+1},\ldots,u_{t+N-1}\}$ denote a sequence of control inputs over a finite horizon of length $N$, and let $x_{t+1:t+N} := \{x_{t+1},x_{t+2},\ldots,x_{t+N}\}$ be the resultant state trajectory. Consequently, we define the cost incurred by this trajectory as
\begin{align}
J & = \sum_{k=t}^{t+N-1}\Big[Q_1\Bigl(\frac{e_{y,k}\!-\!e_{y,\dest}}{e_{y,\max}}\Bigl)^2\!+\! Q_2(\gamma_{k}\!-\!\gamma_{\max})^2\! \nonumber
\\ & \quad +\!R_1\Bigl(\frac{a_k}{a_{\max}}\Bigl)^2\nonumber
  \!+\!R_2\Bigl(\frac{\delta_k}{\delta_{\max}}\Bigl)^2+\!R_3\Bigl(\frac{d_k}{d_{\min}}\Bigl)^2 +\! \Delta R_1\Bigl(\Delta a_k\Bigl)^2 \nonumber
  \\ & \quad +\!\Delta R_2\Bigl(\Delta \delta_k\Bigl)^2 \Big] \!+\! P_1\Bigl(\frac{v_{x,t+N}\!-\!v_{x, \max}}{v_{x,\max}}\Bigl)^2\! \nonumber
  \\ &\quad  \!+\!P_2\Bigl(\frac{e_{y,t+N}\!-\!e_{y,\dest}}{e_{y,\max}}\Bigl)^2 \!+\!P_3(\gamma_{t+N}\!-\!\gamma_{\max})^2, \label{eq:cost_function}
\end{align}
where $Q_i, R_i, \Delta R_i$, and $P_i$ denote the coefficients of the cost function pertaining to cost on the deviation of states from desired state values, cost of control inputs, penalty term to minimize jerk, and terminal cost, respectively. The jerk minimization terms penalize change in control inputs, formally defined as $\Delta u_k:= [\Delta a_k, \Delta \delta_k, \Delta d_k]^\intercal$ where $\Delta a_k= a_k -a_{k-1}$, $\Delta \delta_k=\delta_k-\delta_{k-1}$ and $\Delta d_k=d_k-d_{k-1}$.

In particular, within the summation, the first term minimizes deviation from lane center, the second term increases the cost when the state of energy of the battery $\gamma$ decreases from the maximum allowed value $\gamma_{\max}$, the next three terms penalize large control efforts followed by rest two terms which penalize longitudinal and lateral jerk. The terms in terminal cost penalize deviation from target speed, deviation from lane center, and deviation from $\gamma_{\max}$ respectively.

The resulting finite-horizon constrained optimal control problem can be compactly stated as:
\begin{subequations}\label{eq:MPC_main}
\begin{align}
\min \quad& \sum_{j=t}^{t+N-1} \Big[||x_j-x_{\dest}||^2_Q +||u_j||^2_R +||\Delta u_j||^2_{\Delta R} \nonumber\\&+ \sum_{m=1}^{M}C_s(\xi^{(m)})^2\Big] +||x_{t+N}-x_{\dest}||^2_P, \label{eq:cost}\\
\text{w.r.t.} \quad & x_{t:(t+N)},u_{t:(t+N-1)},\lambda^{(1)}_{1:N}\ldots\lambda^{(M)}_{1:N},\\&\mu^{(1)}_{1:N}\ldots\mu^{(M)}_{1:N} , \xi^{(1)}, \ldots, \xi^{(M)}, \nonumber \\
\text{s.t.}\quad & x_{t+k+1}=x_{t+k}+T_sf(x_{t+k},u_{t+k}), \label{dyn}\\
& u_{\min}\leq u_{t+k} \leq u_{\max}, \label{c11} \\
& x_{\min} \leq x_{t+k} \leq x_{\max}, \label{c1}\\
& -\tau_{\mathrm{m},\max}(x_{t+k})\leq \tau_{\mathrm{m},t+k}\leq \tau_{\mathrm{m},\max}(x_{t+k}), \label{c2}\\ 
& \dd_{\safe}-(A^{(m)}p_{t+k}-b^{(m)})^\intercal \lambda^{(m)}_k \nonumber \\&+\mu_{k}^{(m)^\intercal} h^{{k}^{(m)}}_t\leq\xi^{(m)},\label{safety}\\ 
& \xi^{(m)}\geq0, \quad   \lambda^{(m)}_k\geq0,\quad  ||A^{(m)^\intercal}\lambda^{(m)}_k||_2\leq1, \label{eq}\\
& \mu_{k}^{(m)^\intercal} G^{{k}^{(m)}}_t=\lambda^{(m)^\intercal}_k A^{(m)}, \quad  \mu_{k}^{(m)}\geq0, \label{cs2}\\
& x_t \quad \text{given},
\end{align}
\end{subequations}
where the notation $||v||_W$ in \eqref{eq:cost} denotes quadratic function $v^\intercal W v$, $Q:=\diag(0, Q_1/e_{y,\max}^2, 0, 0, 0, 0, Q_2, 0, 0, 0)$,$R:=\diag(R_1/a_{\max}^2,R_2/\delta_{\max}^2,R_3)$, $\Delta R:=\diag(\Delta R_1,\Delta R_2,0)$  and $P:=\diag(0, P_2/e_{y,\max}^2, 0, P_1/v_{x,\max}^2, 0, 0, P_3, 0, 0, 0)$.  \

The equality constraint representing EV dynamics in \eqref{dyn} holds for $k \in \{0,1,\ldots,N-1\}$, the constraints on input \eqref{c11} encode \eqref{12}-\eqref{13}, the constraint \eqref{c2} encodes \eqref{trq} (and is a function of longitudinal acceleration input $a$), and hold for $k \in \{0,1,\ldots,N-1\}$. Similarly, the bounds on state variables \eqref{c1} encode \eqref{sb1}-\eqref{sb3} and these constraints hold for $k \in [N]$. Finally, the constraints \eqref{safety}-\eqref{cs2} encode robust collision-avoidance constraints derived in Proposition \ref{dynamic} for $M$ dynamic obstacles and hold for all $m \in [M]$ and $k \in [N]$.  

In order to ensure that the MPC formulation remains feasible at all time steps and the algorithm terminates without exceeding the maximum allowed iterations, slack variables $\xi^{(m)}$ are included in the safety constraints (\ref{safety}) for each of the $m$ obstacles. These slack variables are penalized with a very large weight $C_s$ in the cost function (\ref{eq:cost}).

At every sampling instant $t$, the position and orientation of each obstacle are obtained, and the set $\mathbb{O}_t$ is constructed. Then, a set of new samples is collected after computing the deviation of the obstacles from their past positions. Then, the occupancy sets \eqref{eq:uncertainty_set} are recomputed using the PCA approach. Then, the above optimization problem, encoding robust collision avoidance constraints and penalizing energy consumption, is solved. The first component of the control input is applied to the EV, the state and obstacles are allowed to evolve, and the process repeats in a receding horizon manner.

The proposed MPC controller is in contrast with approaches where a path planning algorithm is used to compute a collision-free path, which is then tracked by a tracking controller. Essentially, we combine both trajectory planning and path tracking functionalities in the proposed approach. The proposed approach is ``robust" in the sense that if collision avoidance constraints \eqref{safety}-\eqref{cs2} are satisfied, then the distance between the ego EV and the dynamic obstacle is guaranteed to be larger than the safety margin for all possible perturbations of the obstacle position from the set $\Omega^k_t(W^k_t)$ constructed from past data of obstacle motion. Thus, the constraints derived in this paper are stronger compared to the case where the distance between the EV and a predicted future trajectory of the obstacle is guaranteed to be larger than the safety margin. We further clarify that the proposed approach is distinct from robust MPC, where robustness concerns uncertainty in the plant dynamics.


\section{Simulation Results}\label{results}

In this section, we demonstrate the effect of curved roads on energy
consumption followed by an illustration of the effectiveness of the proposed energy-aware and safe motion planning framework through simulations. The EV model is developed in Matlab 2021b based on reference applications \cite{EV}, \cite{DLC}. The values of various parameters related to vehicle dynamics and cost function in \eqref{eq:MPC_main} used in the simulation are given in Table~\ref{para}, unless stated otherwise.

\begin{table}[htbp]
\centering
\begin{tabular}{|l|r||l|r|}
\hline
Parameters & Values & Parameters & Values\\
\hline
$m$& $1611$ kg & $a_{\max}, a_{\min}$ & $4.5, -4$ $\text{m}/\text{s}^2$\\ 
$I_\mathrm{z}$& $3000$ kg$\text{m}^2$ &$\delta_{\max}, \delta_{\min}$ & $0.5, -0.5$ rad\\
$l_F$&$1.188$ m &$v_{x,\min}$ & 0 $ \text{m}/\text{s}$\\ 
$l_R$&$1.512$ m &$e_{y,\max}, e_{y,\min}$ & $3.5, -3.5 $ m \\ 
$C_{\alpha F}$ & $6.3\times10^4 \text{N}/\text{rad}$ &$d_{\max}, d_{\min}$ & $0, -5.75$ $\text{m}/\text{s}^2$\\
$N_\mathrm{d}$&$7.94$ & $\gamma_{\min}, \gamma_{\max}$ &$0.1, 0.9$ \\ 
$C_{\alpha R}$ & $6.3\times10^4 \text{N}/\text{rad}$ &$N$ & $20$\\
$r_\mathrm{w}$&$0.33$ m &$T_f, T_s$ & $0.01$ s, $0.1$ s\\
$E_\mathrm{b}$&$54.28$ kWh & $\eta_\mathrm{b}$ & $1$\\
$C_\mathrm{d}$&$0.28$ kWh & $A_\mathrm{f}$ & 2.27 $\text{m}^2$\\
$\rho_\mathrm{a}$&1.24 $\text{kg}/\text{m}^3$  & $g$ & 9.8 $\text{m}/\text{s}^2$\\
$C_\mathrm{r}$&0.01 & $\eta_\mathrm{i}$ & 1 \\ \hline
$R$&\multicolumn{3}{|r|}{$\diag(0.00025,0.00125,0.05)$}\\
$\Delta R$&\multicolumn{3}{|r|}{$\diag(0.0025,0.0125,0)$}\\
$P$&\multicolumn{3}{|r|}{$\diag(0,0.0019,0,0.00025/0.0015,0,0,1,0,0,0)$}\\
$Q$&\multicolumn{3}{|r|}{$\diag(0,0.0019,0,0,0,0,1,0,0,0)$}\\
$C_s$&\multicolumn{3}{|r|}{$10^5$}\\\hline
\end{tabular}
\caption{Simulation Parameters} \label{para}
\end{table}

The coefficients for $P_\mathrm{b}$ as a function of motor torque ($\tau_\mathrm{m}$) and motor speed ($\omega_\mathrm{m}$) as stated in \eqref{eq:battery_power_approx}, and the coefficients of the bounds on motor torque $\tau_{\mathrm{m},\max}$ as a function of $v_x$ are obtained through least squares curve fitting. These coefficients and accuracy of fit are given in Tables \ref{tab:Pb_map} and \ref{tab:trq_lim}, respectively. The coefficients $c_3$ and $c_7$ of $P_\mathrm{b}$ are approximately zero and not included in the table.

\begin{table}[htbp]
    \centering
    \begin{tabular}{|c|c|c|c|c|c|}
    \hline
        Coefficients & $c_1$&$c_2$&$c_4$&$c_5$&$c_6$  \\ \hline
         Values & $-1144$ &$0.7604$  & $0.0043$ & $1$&$0.0721$\\ \hline \hline
         R-squared value & \multicolumn{5}{|c|}{0.999}\\ \hline
    \end{tabular}
    \caption{Estimated coefficients of the polynomial approximation of battery output power $P_\mathrm{b}$ in \eqref{eq:battery_power_approx}.}
    \label{tab:Pb_map}
\end{table}

\begin{table}[htbp]
    \centering
    \begin{tabular}{|c|c|c|c|c|}
    \hline
        Coefficients & $\xi_1$&$\xi_2$&$\xi_3$&$\xi_4$  \\ \hline
         Values &$0.0036$&$-0.3661$&$3.663$&$454.2$ \\ \hline \hline
         R-squared value & \multicolumn{4}{|c|}{0.974}\\ \hline
    \end{tabular}
    \caption{Estimated coefficients of the polynomial approximation of motor maximum torque limit $\tau_{\mathrm{m},\max}$ in \eqref{trq}.}
    \label{tab:trq_lim}
\end{table}

The Automated Driving Toolbox of Matlab is used to generate scenarios for evaluating the performance of the proposed framework.  At every time instant, the uncertainty set parameters $G$'s and $h$'s are obtained using Multi Parametric Toolbox \cite{MPT3}. Nonlinear solver IPOPT \cite{wachter2006implementation} is used to solve the optimization problem in Python using Casadi modelling language \cite{Andersson2019}. All computations were carried out on a Workstation Desktop with $128$ GB RAM and Intel Xeon(R) Gold 6146 CPU @ $3.20\text{GHz}\times48$ Processor. The computation was carried out in single-threaded mode using only one core of the CPU. A block diagram of the closed-loop system implementation under the proposed approach is shown in  Figure \ref{fig:arch}. The collision avoidance constraints are active when the obstacle is within a range of $100$ m and at least $N_s$ samples of obstacle displacement data have been gathered. When the specified requirements are not true, collision avoidance constraint is inactive and dropped from the optimization problem with the value of variable $\dd_{\EO}$ (which represents the distance between ego EV and obstacle) set to $0$ in the figures. 

\begin{figure}
    \centering
    \begin{tikzpicture}[impt/.style={rectangle, draw=red, minimum size=5mm},]
    \node[impt] (a) {Plant \& Environment};
    \node[impt] (gen) [below=1cm of a]{Uncertainty Set Generator};
    \node[impt] (cont) [left=2cm of a]{MPC};
    \draw[->, thick] (cont.east) to node[above]{$u_t$} (a.west);
    \draw[->,thick](gen.west) -| node[near start,above]{$(G^k_t,h^k_t)$} (cont.south);
    \draw[->,thick](a.north) |-(0,7.5mm)   -| node[near start,above]{$x_t$} (cont.north);
    \draw[->,thick](a.south) to node[midway,right]{$W^k_t$} (gen.north);
    \node[draw, thick, blue,densely dashed, fit=(a) (gen), label=above right:\textcolor{blue}{Matlab}] (box) {};
    \node[draw, thick, blue,densely dashed, fit=(cont), label=below right:\textcolor{blue}{Python}] (box) {};
\end{tikzpicture}
    \caption{Implementation architecture. Here, the index $k\in [N]$.}
    \label{fig:arch}
\end{figure}


The ego EV plans and follows a trajectory that avoids static as well as dynamic obstacles while trying to keep the level of stored energy $\gamma$ close to $\gamma_{\max} = 0.9$. We have chosen scenario sampling time $T_f=0.01$s, controller sampling time $T_s = 0.1$s, and prediction horizon $N = 20$ for the MPC formulation \eqref{eq:MPC_main}. Note that, in addition to the motion planning controller, a motion-tracking controller operating at a significantly higher frequency is required for real-time implementation of the proposed scheme. However, the focus of this paper is on the motion planning aspect, and the lower-level control aspects associated with the motion-tracking controller are not addressed here.
 
We now define different components of energy consumption and recovery for an EV. The battery energy consumption is denoted by $E_B$. The forces at the wheels are responsible for maneuvering the EV over a specified duration  $[t_0,t_f]$. The traction energy consumed from the powertrain over a planar trajectory is denoted by $E_T$, defined as
\begin{align}
E_T&\coloneqq\int_{t_0}^{t_f} \tau_{\mathrm{m}}\omega_\mathrm{m} \,dt=\int_{t_0}^{t_f} P_T \,dt.    
\end{align}
It is emphasized that the EV model given in (\ref{1}-\ref{1_end}) assumes the wheel camber angle to be zero and ignores wheel slippage loss, friction loss at the wheel, and rotational energy consumed at the wheel. Therefore, 
\begin{align}
    E_T&= \int_{t_0}^{t_f} F_\mathrm{x} v_x \,dt.
\end{align}
A part of $E_T$ is consumed (dissipated) by  lateral forces at wheels and is denoted by $E_{WL}$ where,
\begin{align}
  E_{WL}&\coloneqq \int_{t_0}^{t_f} (F_{F,y}v_{F,y}^w+F_{R,y}v_{R,y}^w) \,dt=\int_{t_0}^{t_f} P_{WL} \,dt. 
\end{align}
Here, $v_{F,y}^w$ ($v_{R,y}^w$) is a component of wheel velocity along the transverse axis of the front (rear) wheel. In the vehicle body coordinate frame, a component of defined $E_{WL}$ leads to a difference between traction energy from the powertrain consumed over the planar maneuver and traction energy consumed over the longitudinal maneuver. Here, the latter is denote by $E_{Long}$ and is defined as
\begin{align}
    E_{Long}&\coloneqq \int_{t_0}^{t_f} (m\dot{v}_xv_x+F_\mathrm{a}v_x+F_\mathrm{r}v_x) \,dt=\int_{t_0}^{t_f} P_{Long} \,dt.
\end{align}
The stated difference is generally not considered in the literature due to the assumption of EV motion along a straight road.
The remaining component of $E_{WL}$ contributes to traction energy consumed (dissipated) over lateral maneuver of EV, which is denoted by $E_{Lat}$ and defined as
\begin{align}
    E_{Lat}&\coloneqq\int_{t_0}^{t_f} (m\dot{v}_yv_y+I_\mathrm{z}\dot{r}r) \,dt=\int_{t_0}^{t_f} P_{Lat} \,dt.
\end{align}
We denote energy consumed (dissipated) by friction braking at wheel over the planar maneuver as $E_{Brake}$. It can be shown by coordination transformation from wheel to vehicle body frame \cite{latdyn} that
\begin{align}
   E_T+E_{WL}+E_{Brake}\approx E_{Long}+E_{Lat}. \label{energy_rel}
\end{align}
 We denote traction power consumed over planar maneuver and total power consumed over EV planar
maneuver as $P_{Man}$ and $P_{Total}$ respectively. Here,
\begin{align}
     P_{Man}&=P_{T}+P_{WL},\\
   P_{Total}&=P_{Man}+P_{Brake}.
\end{align}

\subsection{Effect of Road Curvature on Energy Consumption}

We now demonstrate the effect of lateral dynamics on energy consumption of EVs on benchmark drive cycles along a curved road. An EV model is simulated to track the FTP-75 drive cycle along a road characterized by different radii of curvature, with speed limited by a maximum speed limit \cite{irc},\cite{amata2008abrupt}. This is a hypothetical case, but mend to convince the reader of the extent of energy consumption in lateral motion since such results are not found in literature.

\begin{figure}[htbp]
    \centering
    \includegraphics[trim={0 0 0 0cm},clip,width=0.45\textwidth]{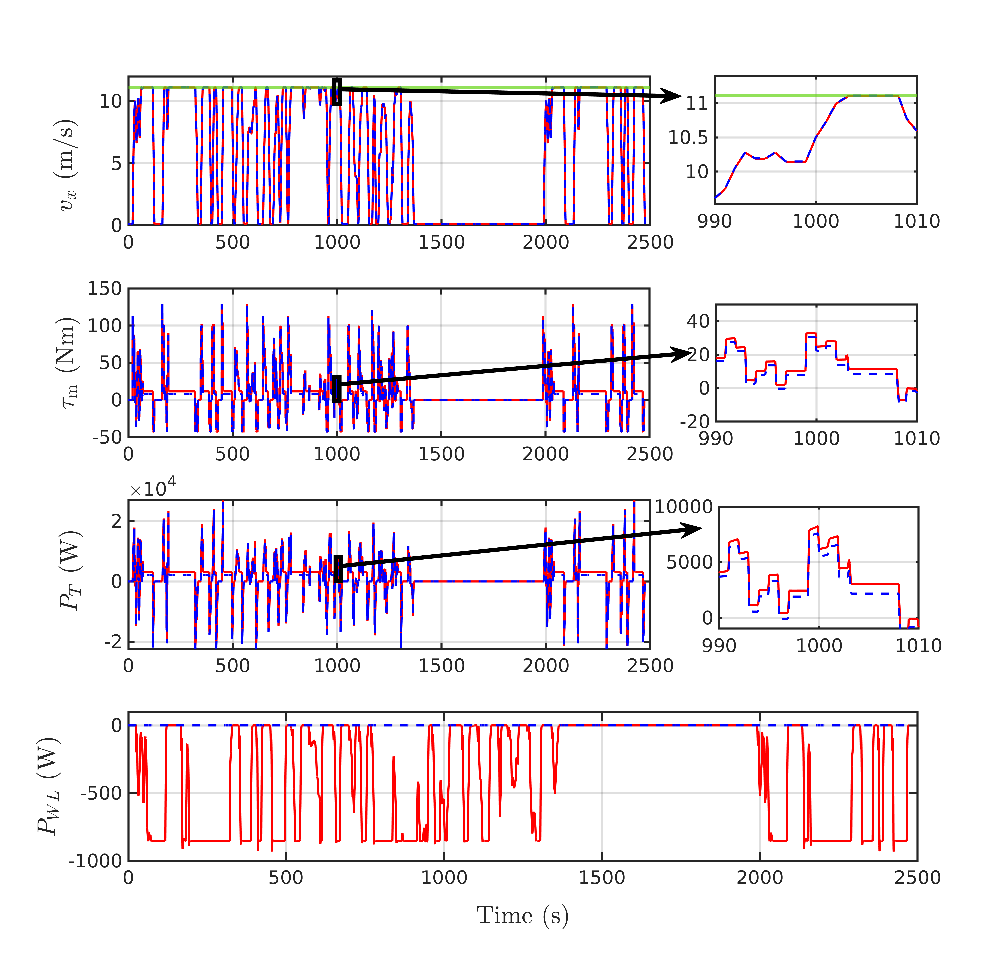}
    \caption{\footnotesize Longitudinal speed ($v_x$), motor torque ($\tau_{\mathrm{m}}$), traction power ($P_T$) and power consumed by wheel lateral forces ($P_{WL}$) profiles of an EV that maneuvers along a straight road (blue dashed) and a curved road with 60 m radius of curvature (red solid).}
    \label{fig:radius60plot}
\end{figure}

\begin{table}[htbp]
\caption{Energy consumption and estimated range for different radius of curvature the road for an EV with $54.28$ kWh battery.}
\begin{tabular}{|p{1.8cm}|p{0.65cm}|p{0.65cm}|p{0.65cm}|p{0.65cm}|p{0.65cm}|p{0.65cm}|}
\hline
\textbf{Parameters}&\multicolumn{3}{|c|}{\textbf{Straight road}}&\multicolumn{3}{|c|}{\textbf{Curved road}}\\
\hline
Drive Cycle & I & II & III & I & II & III \\
\hline
Radius (m) & $\infty$& $\infty$& $\infty$& 60& 250 &500 \\
\hline
Distance (km)& 14 & 17.36 & 17.86 & 14 & 17.36 & 17.86\\ \hline
Battery Energy (kWh) & $0.99$ & $1.52$ &$1.64$& $1.26$ & $1.62$ &$1.69$ \\ \hline
Estimated Range (km) & $766.2$& $621.1$&$589.9$ & $602.8$ & $580.5$ &$575.5$ \\ \hline
\end{tabular}
\label{tab_curve}
\end{table}
Figure \ref{fig:radius60plot} illustrates the occurrence of significant power dissipation due to the lateral motion of the wheel on a curved route as opposed to a straight one when following a particular drive cycle. As a result, the estimation of the driving range of EV without taking the radius of curvature of the road into account is exaggerated as shown in Table \ref{tab_curve}. 

Additionally, we validate this observation by simulating the EV model to track real-world recorded drive cycle data extracted from the Next Generation Simulation (NGSIM) highway driving dataset \cite{huang2021driving}. The power dissipation profiles due to the lateral motion of the wheel on a curved road (with a radius of curvature $250$ m) and a straight road for this realistic speed profile are shown in Fig.~\ref{fig:radius250plot}. The corresponding estimated driving range for the EV is given in Table~\ref{tab_curve_250}. 

In both studies, we observe a reduction in driving range when an EV tracks the same drive cycle on a curved road as opposed to a straight road. Additionally, the difference in estimated range increases with a decrease in the radius of curvature. Therefore, it is essential to consider the effect of lateral motion on energy consumption. 

\begin{figure}[htbp]
    \centering
    \includegraphics[trim={0 0 0 0cm},clip,width=0.45\textwidth]{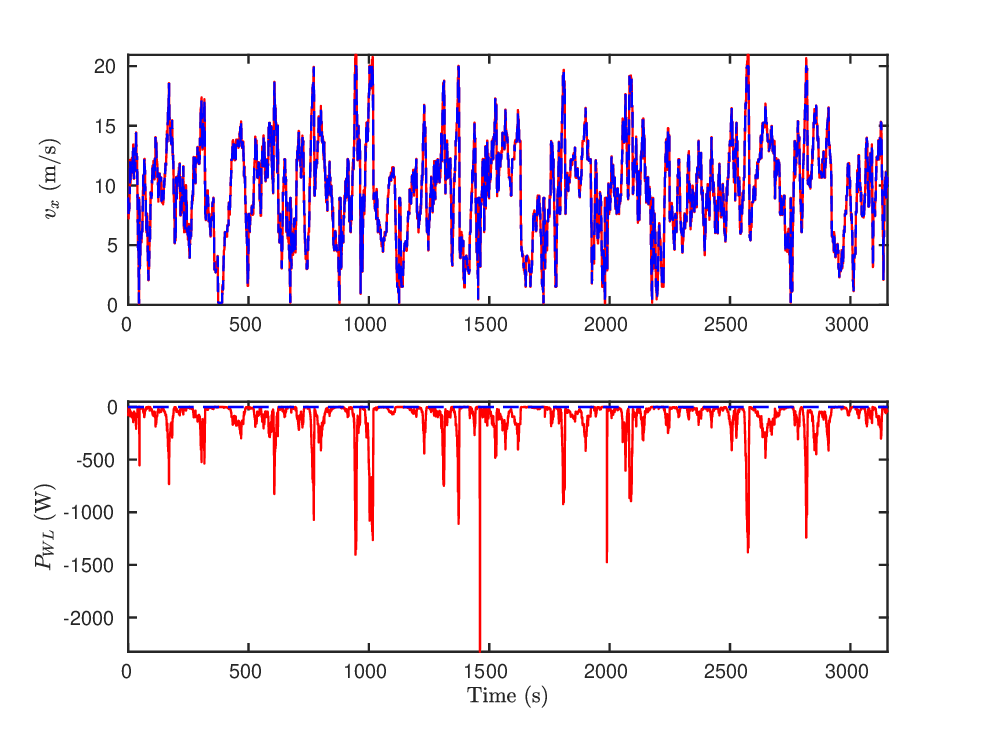}
    \caption{\footnotesize Longitudinal speed ($v_x$) and power consumed by wheel lateral forces ($P_{WL}$) profiles of an EV that maneuvers along a straight road (blue dashed) and a curved road with 250 m radius of curvature (red solid).}
    \label{fig:radius250plot}
\end{figure}

\begin{table}[htbp]
\caption{Energy consumption and estimated range on a straight and a curved road for an EV tracking drive cycle extracted from the NGSIM dataset with a $54.28$ kWh battery.}
\centering
\begin{tabular}{|p{3cm}|p{2cm}|p{2cm}|}
\hline
\textbf{Parameters}&\textbf{Straight road}&\textbf{Curved road}\\
\hline
Radius (m) & $\infty$& 250 \\
\hline
Distance (km)& 30.56& 30.57\\ \hline
Battery Energy (kWh) & $3.19$  &$3.29$ \\ \hline
Estimated Range (km) & $519.5$& $504.1$ \\ \hline
\end{tabular}
\label{tab_curve_250}
\end{table}

Next, we evaluate the effectiveness of the proposed energy-aware safe motion planning framework via simulations for different scenarios, including motion on a straight road, a curved road, multiple obstacles, and real-world traffic.


\begin{figure*}[t]
    \centering
    \begin{subfigure}[b]{0.32\textwidth}
    \includegraphics[trim={0 0cm 0 0cm},clip,width=\textwidth,height=6.5cm]{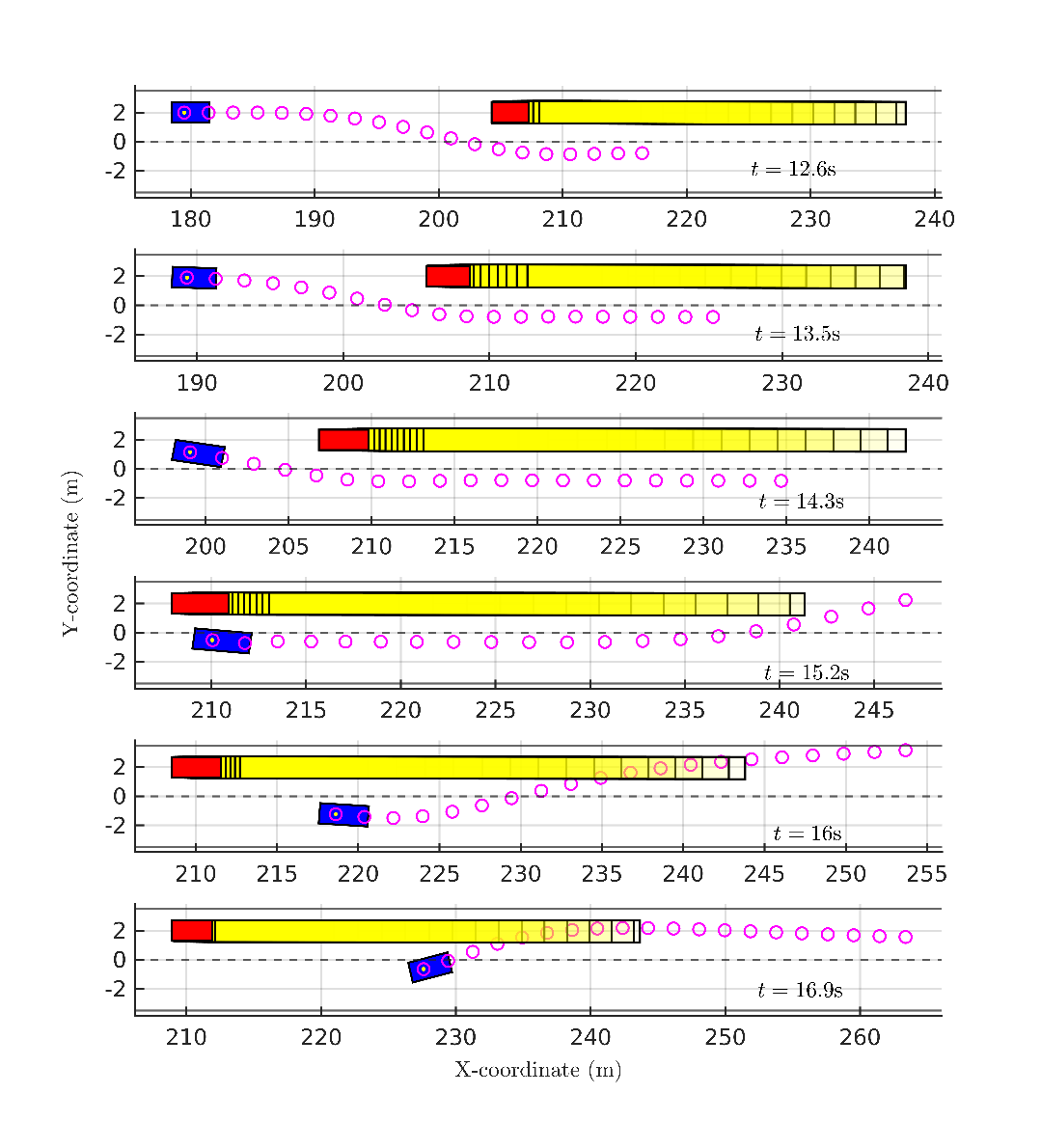}
    \caption{\footnotesize Evolution of vehicle position in the EU case.}
\label{traj_ovtk_straight}
    \end{subfigure}
    \hfill
      \begin{subfigure}[b]{0.32\textwidth}
 \includegraphics[trim={0 0cm 0 0cm},clip,width=\textwidth,height=6.5cm]{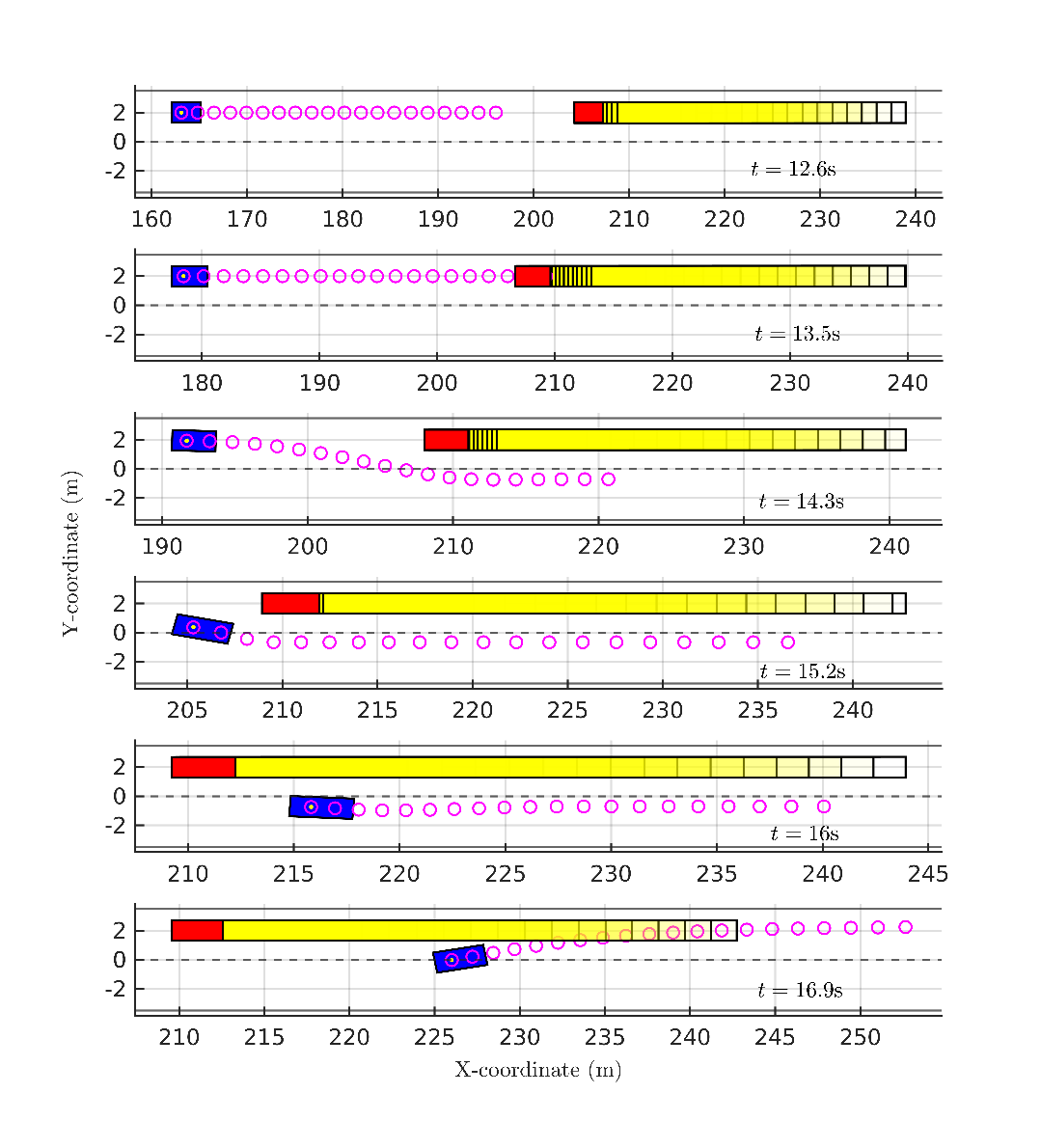}
 \caption{\footnotesize Evolution of vehicle position in the EA case.}
\label{traj_ovtk_soe_straight}
    \end{subfigure} 
   \hfill
      \begin{subfigure}[b]{0.32\textwidth}
 \includegraphics[trim={0 0cm 0 0cm},clip,width=\textwidth,height=6.5cm]{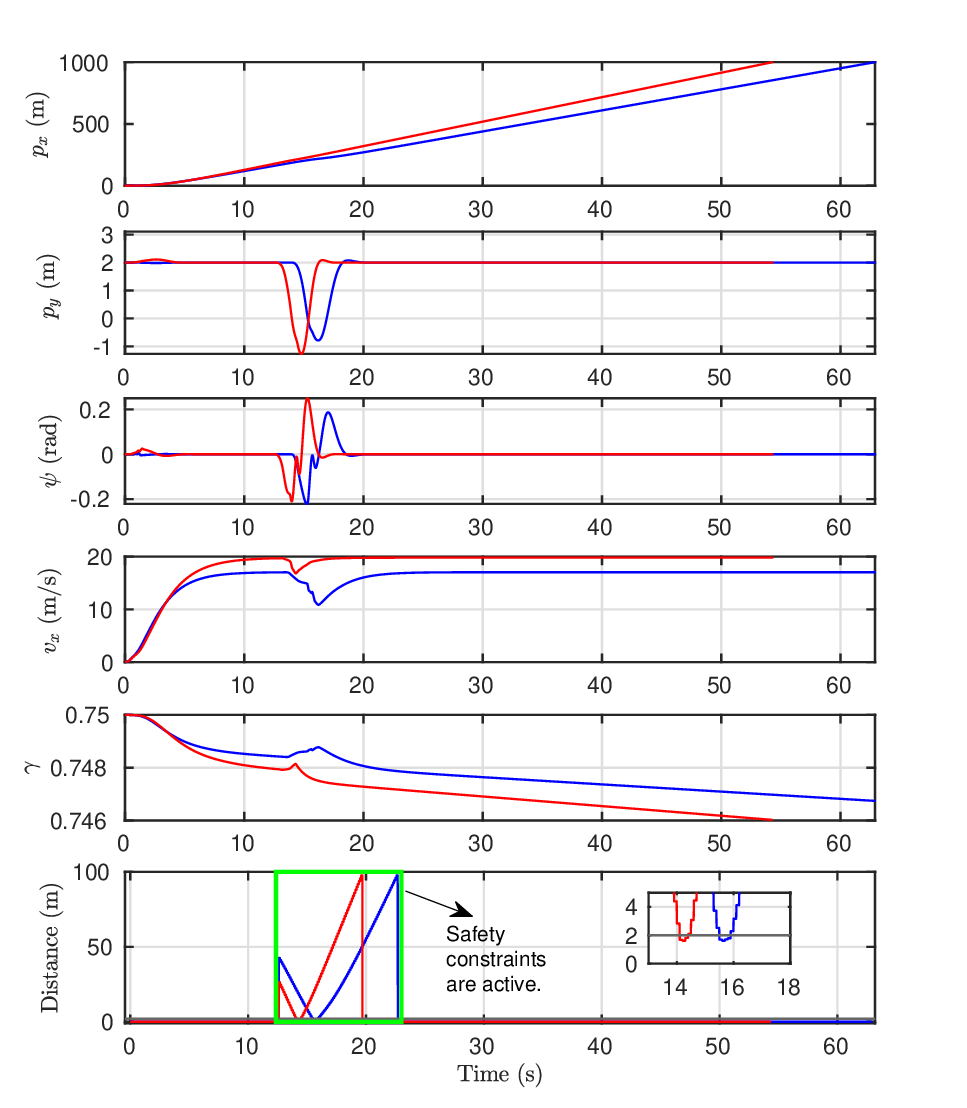}
 \caption{\footnotesize Ego EV states and distance from obstacle.}
\label{state_ovtk_straight}
    \end{subfigure} 
    \caption{\footnotesize For (a) and (b), subplots show ego EV trajectory evolution without and with energy consumption consideration respectively at different time instants. Magenta: Trajectory of ego EV, Red: Obstacle, Yellow (increasing transparency with increase in time): Future occupancy set. (c) Ego EV state evolution with time. Red: energy consumption is not penalized and Blue: energy consumption is penalized. In last subplot, Grey: $\dd_{\safe}$, Blue/Red: $\dd_{\EO}$.}
    \label{fig:strght_ovtk}
\end{figure*}

\subsection{Efficacy of {\it Energy-Aware} Safe Motion Planning Approach}
We compare the energy consumption (including traction energy consumed over lateral and longitudinal maneuvers introduced earlier), and time taken to reach the destination under the proposed {\it Energy-Aware} (EA) scheme with another MPC-based controller that does not penalize energy in its cost function (i.e., without the $(\gamma_k - \gamma_{\max})^2$ term), but maintains the collision avoidance constraints to ensure a fair comparison. The latter scheme is termed {\it Energy-Unaware} (EU) scheme. The penalty term $(\gamma_k - \gamma_{\max})^2$ present in the EA approach optimizes $\tau_\mathrm{m}$ and $\omega_\mathrm{m}$ in order to minimize energy consumption. These terms further depend on acceleration and steering inputs applied to EV. In addition to acceleration input, steering input is also optimized so that less resistance is applied by lateral forces at wheels along the longitudinal axis of the EV. This leads to a reduction in traction energy consumed as well as an increase in traction energy recovered for the EA approach compared to the EU approach when EV tracks same longitudinal speed for both. In addition to this, a reduction in resistance also results in smaller power dissipation along the transverse axis of the EV as a by-product.

We now discuss the efficacy of the proposed approach for different driving scenarios.

\subsubsection{Overtaking on a straight road} 
In this scenario, the obstacle is moving in front of the ego EV in the same lane. It is located 150m ahead initially, starts from rest, accelerates to some speed, and then decelerates to stop after traveling a distance of 240m. This could potentially represent a moving obstacle that slows down to cause a roadblock on a particular lane. The target or destination state, $x_{\dest}$, for the ego EV is set to $[0,2,0,20,0,0,0.9,0,0,0]^\intercal$ in this case. The model is simulated till the ego EV travels up to $1$ km. The value of $N_s$ is chosen to be $1250$ to ensure $99\%$ confidence ($\beta=0.01$) for actual change in obstacle position to lie in $\Omega_t^*(W_t^*)$ with the probability of $0.99$ ($\epsilon=0.01$) in accordance with Theorem~\ref{thrm1}. This description also applies to the curved road scenario presented afterwards.

Figure \ref{traj_ovtk_straight} and \ref{traj_ovtk_soe_straight} depict the ego EV trajectories for the EU and EA schemes, respectively. It is evident that the ego EV overtakes the obstacle (which is slowing down) while avoiding collision. Figure \ref{state_ovtk_straight} shows the evolution of ego EV states for both cases. The ego EV performs an aggressive overtake under the EU scheme as indicated by variation in yaw rate at high speed. The ego EV maneuver employs lesser deceleration followed by a lesser amount of acceleration during overtake for EU scheme compared to the EA one. For the EA scheme, the ego EV consumes a smaller proportion of maneuver energy as well as recuperates a relatively large amount of $E_T$ during deceleration, resulting in a lower overall $E_T$ compared to the EU scheme. The amount of $E_T$ under the EA scheme for the EV traveling up to 1 km with a single overtaking along a straight road is $0.1529$ kWh, however, the energy consumed is $0.1911$ kWh under the EU scheme. 

 \begin{figure}[htbp]
    \centering
    \includegraphics[trim={0 0 0 0cm},clip,width=0.45\textwidth]{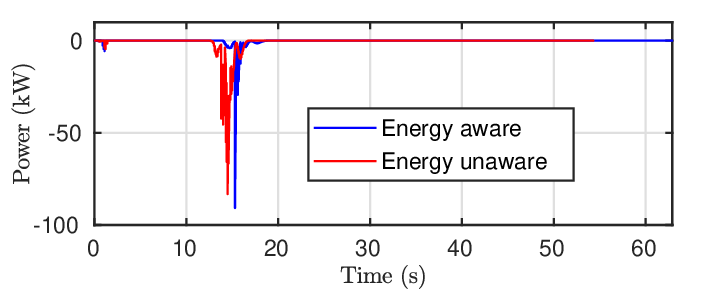}
    \caption{\footnotesize Power consumed by wheel lateral forces ($P_{WL}$) profile for ego EV overtake maneuver given in Figure \ref{fig:strght_ovtk}.}
    \label{fig:ypwr}
\end{figure}
\begin{figure}[htbp]
    \centering
    \includegraphics[trim={0 0 0 0cm},clip,width=0.45\textwidth]{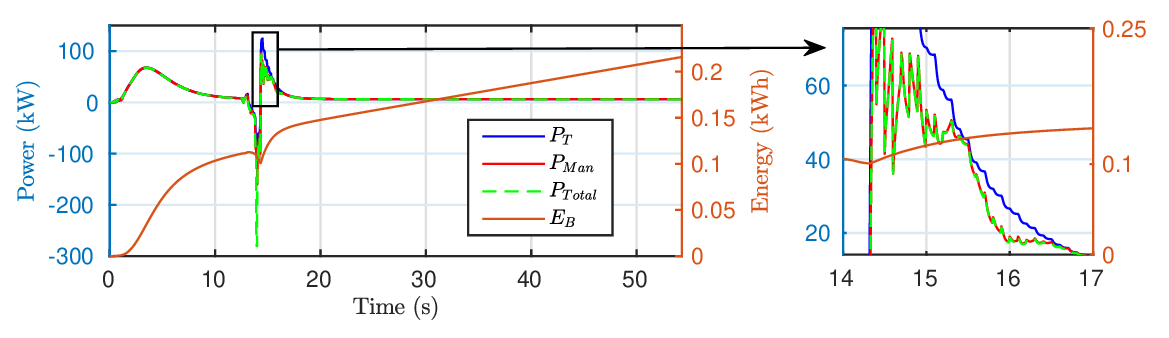}
    \includegraphics[trim={0 0 0 0cm},clip,width=0.45\textwidth]{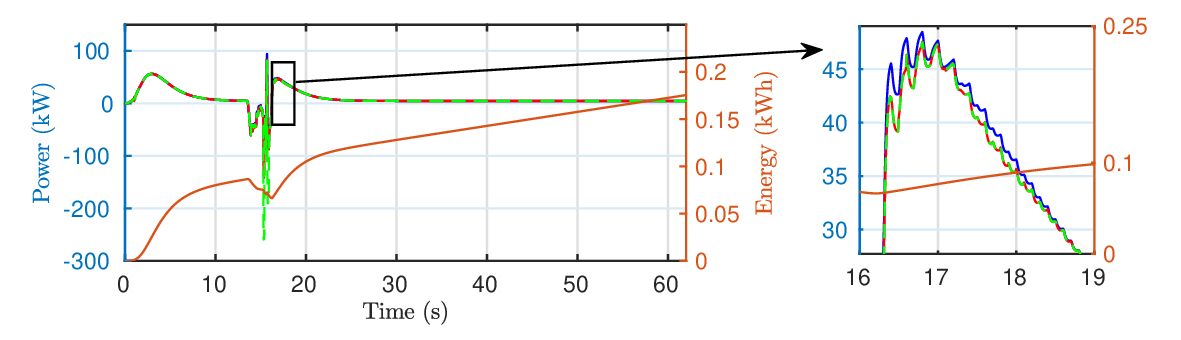}
    \caption{\footnotesize Power flow profile for EU (top) and EA (bottom) overtake maneuver given in Figure \ref{fig:strght_ovtk}.}
    \label{fig:pwr_ovtk}
\end{figure}


\begin{figure*}[t]
    \centering
    \begin{subfigure}[b]{0.32\textwidth}
    \includegraphics[trim={0 0cm 0 0cm},clip,width=\textwidth,height=7cm]{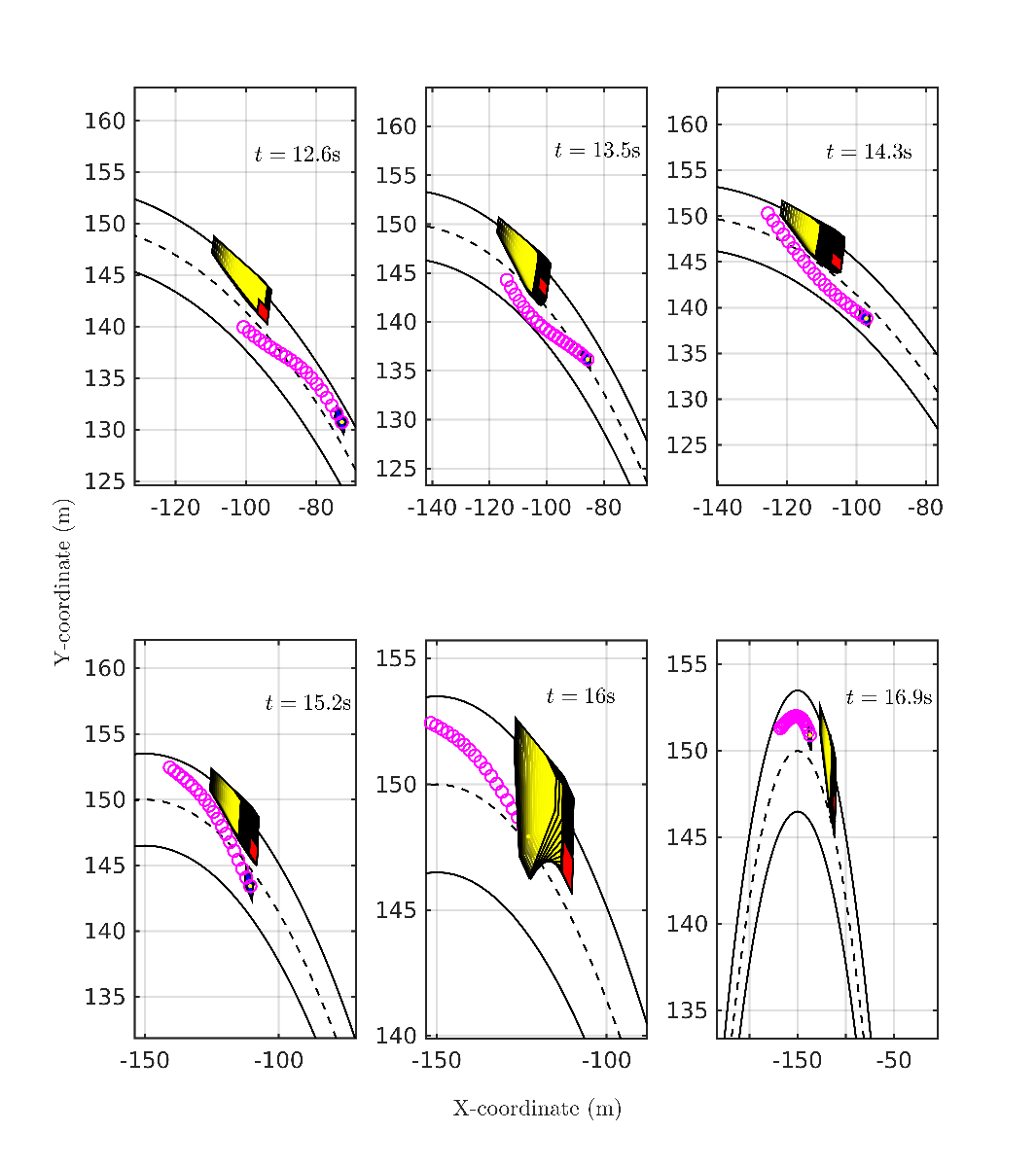}
    \caption{\footnotesize Evolution of vehicle position in the EU case.}
\label{traj_ovtk_curve_1}
    \end{subfigure}
    \hfill
      \begin{subfigure}[b]{0.32\textwidth}
 \includegraphics[trim={0 0cm 0 0cm},clip,width=\textwidth,height=7cm]{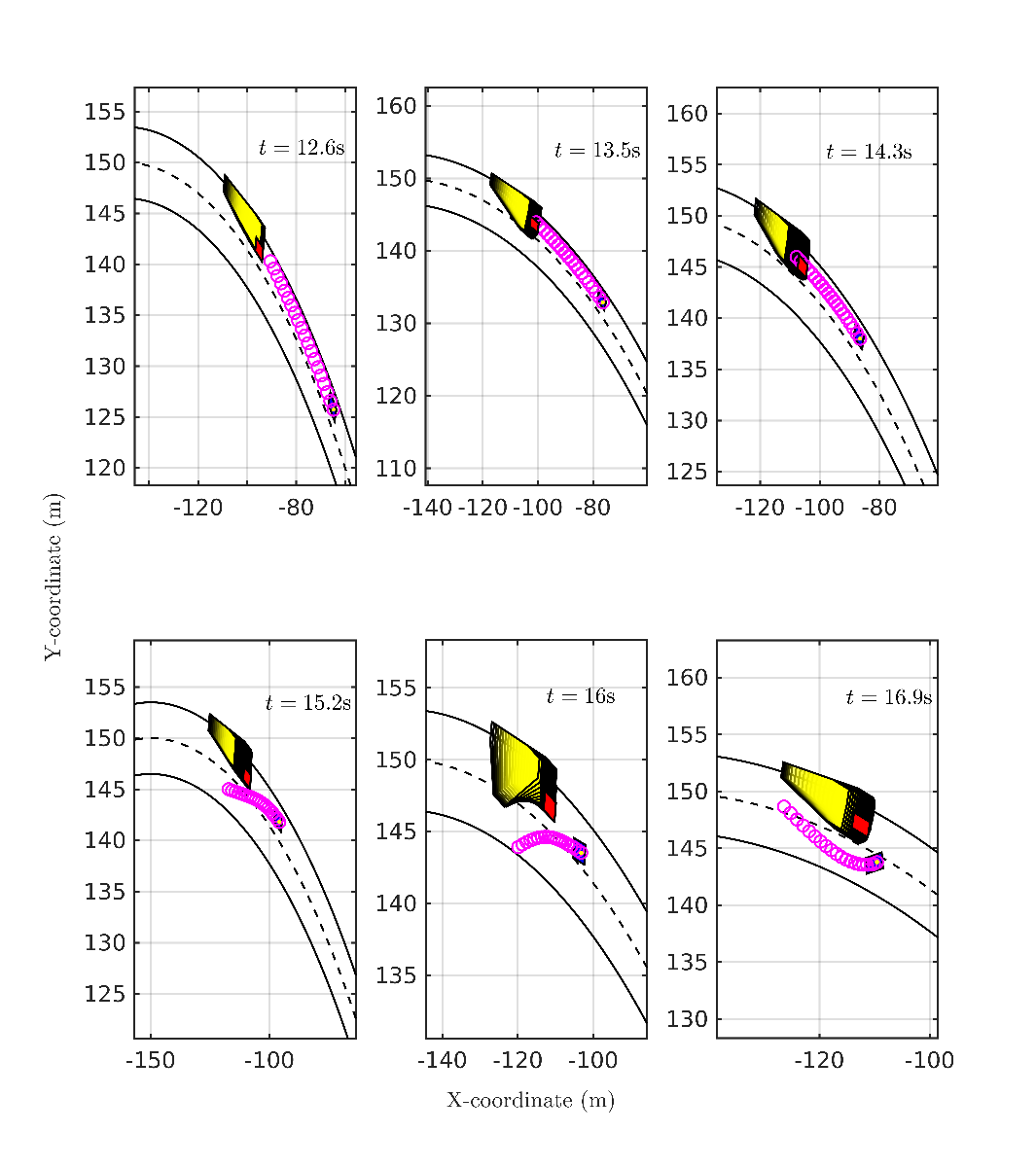}
 \caption{\footnotesize Evolution of vehicle position in the EA case.}
\label{traj_ovtk_soe_curve_1}
    \end{subfigure} 
         \hfill
      \begin{subfigure}[b]{0.32\textwidth}
 \includegraphics[trim={0 0cm 0 0cm},clip,width=\textwidth,height=7cm]{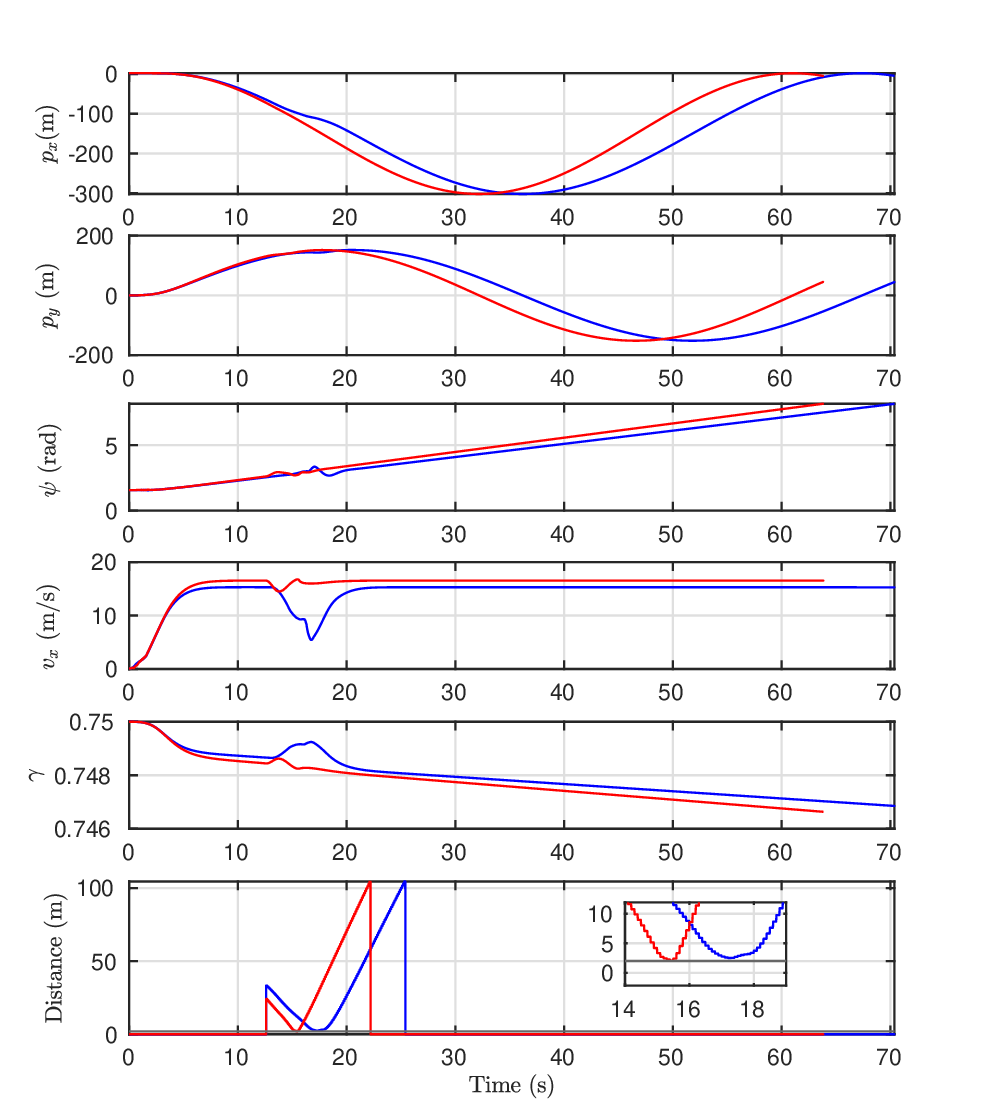}
 \caption{\footnotesize Ego EV states and distance from the obstacle.}
\label{state_ovtk_curve}
    \end{subfigure}  
    \caption{\footnotesize For (a) and (b), subplots show ego EV trajectory evolution without and with energy consumption consideration respectively at different time instants. Magenta: Trajectory of ego EV, Red: Obstacle, Yellow (increasing transparency with increase in time): Future occupancy set. (c) Ego EV state evolution with time. Red: energy consumption is not penalized and Blue: energy consumption is penalized. In last subplot, Grey: $\dd_{\safe}$, Blue/Red: $\dd_{\EO}$.}
    \label{fig:curve_ovtk}
\end{figure*}

It can be observed from Figure \ref{fig:ypwr} that $P_{WL}$ has substantial nonzero magnitude during the overtaking period, resulting in $0.0146$ kWh of stated $E_T$ lost to wheel lateral forces for the EU scheme, which is larger compared to $0.0052$ kWh for the EA scheme. This dissipation results in only a fraction of the $E_T$ being consumed in the vehicle longitudinal maneuver as $E_{Long}$, and the other being dissipated in the lateral maneuver as $E_{Lat}$; the exact quantities are $0.1696$ kWh and $0.00047$ kWh for the EU case, and $0.1386$ kWh and $7\times 10^{-5}$ kWh for the EA case. The corresponding power flow and battery energy consumption ($E_B$) profile for both schemes is given in Figure \ref{fig:pwr_ovtk}, and the detailed breakdown of energy consumption is given in Table \ref{tab:energy_breakup}.

The vehicle covered the distance of $1$ km in $54.3$ seconds in the EU scheme as it performed aggressive maneuver, while it took $62.9$ seconds in the EA case where the ego EV performed smooth longitudinal and lateral maneuvers to reduce energy consumption. Thus, the proposed scheme was able to achieve a significant reduction in energy consumption with an increase in travel time of $8.6$ seconds. This observation is consistent with those reported in the literature \cite{kamal2012model,he2020multiobjective,hadjigeorgiou2019optimizing}. The variation of $\dd_{\EO}$ is given in the last subplot of Figure \ref{state_ovtk_straight} which shows that the proposed approach optimizes continuously to keep $\dd_{\EO}$ greater than $\dd_{\safe}$ (chosen to be $2$m) \cite{hussain2018real},\cite{nilsson2015longitudinal}, except for a very small duration where the slack variable was active to obtain a feasible solution within a reasonable time. In practice, one should choose the value of $\dd_{\safe}$ in the MPC optimization problem to be slightly larger than the actual safety requirement to avoid any safety issues caused by the use of slack variables. Similarly, parallel processing will potentially help in mitigating this issue as one can choose a larger prediction horizon and/or obtain faster convergence to the optimal solution. Depending on the vehicle and obstacle type, a different value of $\dd_{\safe}$ may be used.

\begin{figure}[htbp]
    \centering
    \includegraphics[trim={0 0 0 0cm},clip,width=0.45\textwidth]{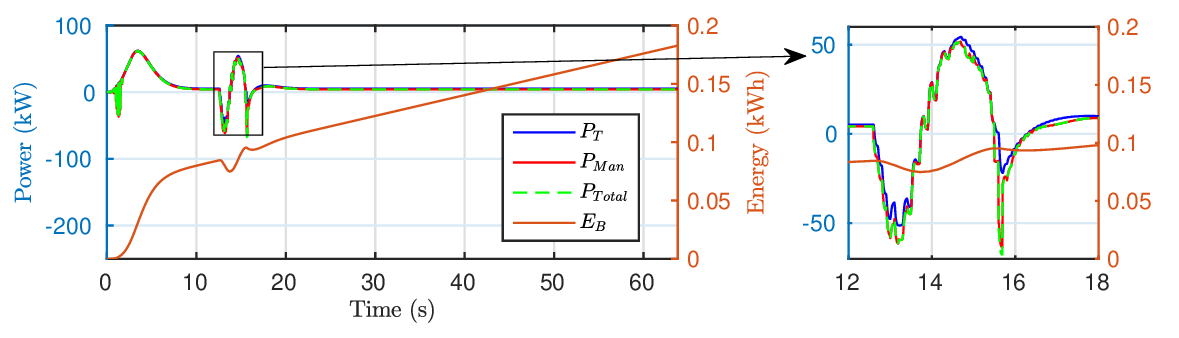}
    \includegraphics[trim={0 0 0 0cm},clip,width=0.45\textwidth]{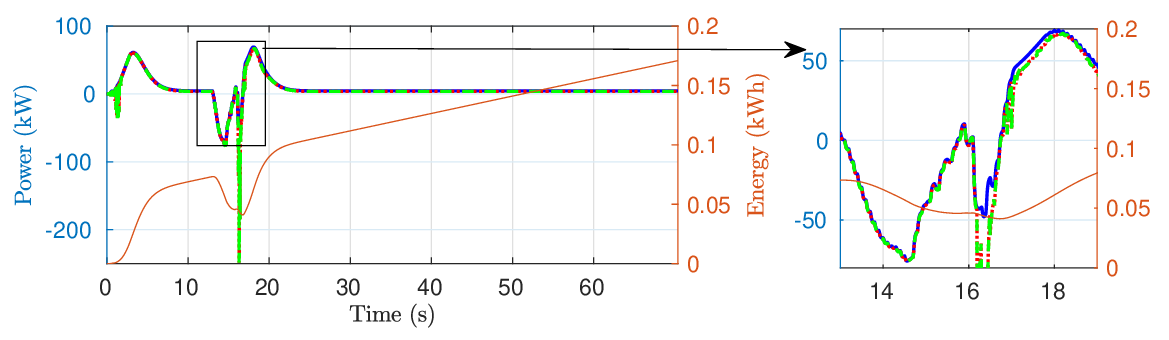}
    \caption{\footnotesize Power flow profile for EU (top) and EA (bottom) overtake maneuver given in Figure \ref{fig:curve_ovtk}.}
    \label{fig:pwr_ovtk_curve}
\end{figure}

\begin{table*}[ht]
        \centering
        \begin{tabular}{|c|p{6.75cm}|c|c|c|c|}
        \hline &
        &\multicolumn{2}{|c|}{\textbf{Straight road}}& \multicolumn{2}{|c|}{\textbf{Curved road}}\\
        \cline{3-6}
    \textbf{S.No.} & \textbf{Variables}&EU case& EA case&EU case& EA case\\ \hline
    \multicolumn{6}{|c|}{\textbf{During Acceleration}}\\ \hline
           1&Battery energy consumed (kWh)   & 0.2279& 0.1993& 0.1943& 0.2036\\ \hline
            2&Traction energy consumed at the wheel over planar maneuver (kWh) ($E_T$)  & 0.2040& 0.1772& 0.1602& 0.1801\\ \hline
            3& Traction energy consumed by wheel lateral forces (kWh) ($E_{WL}$)& -0.0102& -0.0013& -0.0183&-0.0174\\ \hline 
            4& Traction energy consumed over vehicle longitudinal maneuver (kWh) ($E_{Long}$)& 0.1914& 0.1750& 0.1420& 0.1637\\ \hline
            5& Traction energy consumed over vehicle lateral maneuver (traverse + yaw) (kWh) ($E_{Lat}$)& -0.0003& $-3.3 \cdot 10^{-5}$& -0.00013& -0.0009\\ \hline
             \multicolumn{6}{|c|}{\textbf{During Deceleration}}\\ \hline
           6&  Battery energy recovered (kWh) & 0.0120& 0.0225&  0.0115& 0.0327\\ \hline
           7&  Traction energy consumed at wheel over planar maneuver (kWh) ($E_T$) & -0.0129& -0.0243&  -0.0116&-0.0360\\ \hline
           8&  Energy consumption due to friction braking (kWh) ($E_{Brake}$) & -0.0043& -0.0091&  -0.00017& -0.00004\\ \hline
           9&  Traction energy consumed by wheel lateral forces (kWh) ($E_{WL}$)& -0.0044& -0.0039& -0.0037&-0.0074\\ \hline
            10& Traction energy consumed over vehicle longitudinal maneuver (kWh) ($E_{Long}$)& -0.0218& -0.0364& -0.0154& -0.0438\\ \hline
             11& Traction energy consumed over vehicle lateral maneuver (traverse + yaw) (kWh) ($E_{Lat}$)& -0.00017& $-3.7\cdot 10^{-5}$& -0.0001& 0.0004\\ \hline
             \hline
            12& Total battery energy consumed (kWh) ($E_B$)& 0.2159 & 0.1768 &0.1828 & 0.1709\\\hline
             13&Estimated Range (km) & 251.46& 306.98& 296.93& 317.61\\ \hline             
        \end{tabular}
        \caption{Summary of different components of energy consumption for the overtaking scenarios presented in Figures \ref{fig:strght_ovtk} and \ref{fig:curve_ovtk} for the EA and EU schemes, respectively. The curved road has a radius of curvature of $150$ meters. Here, negative values indicate energy recuperation for $E_{T}$, $E_{Long}$, and energy dissipation for $E_{WL}$, $E_{Lat}$, $E_{Brake}$. The energy values satisfy the energy balance equation \eqref{energy_rel} after combining them over the acceleration and deceleration durations.}
        \label{tab:energy_breakup}
    \end{table*}

\subsubsection{Overtaking on a curved road}

Here, the target or destination state, $x_{\dest}$, for the ego EV is set to $[0,2,0,16.67,0,0,0.9,0,0,0]^\intercal$ and the curved road has a radius of curvature of $150$ m. The trajectories of ego EV under both EU and EA schemes are shown in Figures \ref{traj_ovtk_curve_1} and \ref{traj_ovtk_soe_curve_1}, respectively. The ego EV overtakes the obstacle (which is slowing down) while avoiding collision. Figure \ref{state_ovtk_curve} shows the evolution of EV states for both cases. For the EA scheme, the ego EV consumes a significant proportion of maneuver energy, but during deceleration, a relatively large amount of $E_T$ is recuperated, resulting in a lower overall $E_T$ compared to the EU scheme. The traction energy ($E_T$) consumed under the EU scheme for the EV is $0.1448$ kWh is larger compared to $0.1441$ kWh for the EA scheme. As seen in Figure \ref{fig:ypwr_curve}, this situation differs from the preceding one as there is a nonzero $P_{WL}$ over the entire maneuver duration because of the curvature of the road. The power flow profiles for both EA and EU schemes are given in Figure \ref{fig:pwr_ovtk_curve} and the detailed breakup of energy consumption is given in Table \ref{tab:energy_breakup}. During simultaneous steering and regenerative braking, some of the traction power consumed during planar maneuvers ($P_{Man}$) is lost due to resistive lateral forces at wheels. Consequently, only a fraction of $P_{Man}$ is available for recuperation. This is demonstrated by the difference between $P_{Man}$ and $P_{T}$ profiles for negative values in Figure \ref{fig:pwr_ovtk_curve}.

In this case, the EV covers the distance of 1km in $63.84$ seconds for the EU and 70.37 seconds for the EA scheme. The proposed scheme can achieve a significant reduction in energy consumption without a significant decrease in travel time, showing the efficacy of this approach in comparison to an EU one. The variation of $\dd_{\EO}$ is given in the last subplot of Figure \ref{state_ovtk_curve} which shows that $\dd_{\EO}$ is greater than $\dd_{\safe}$ (chosen to be $2$m) throughout the duration of overtaking.

\begin{figure}[htbp]
    \centering
    \includegraphics[trim={0 0 0 0cm},clip,width=0.4\textwidth]{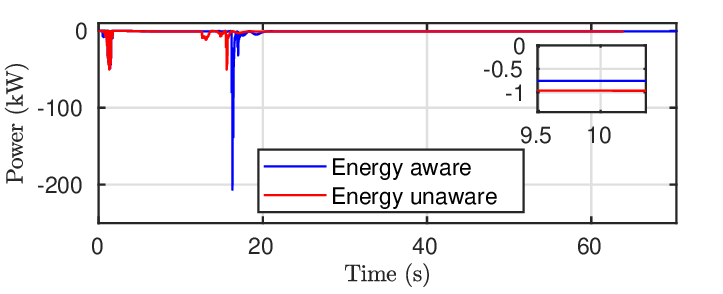}
    \caption{\footnotesize Power consumed by wheel lateral forces ($P_{WL}$) profile for ego EV overtake maneuver given in Figure \ref{fig:curve_ovtk}.}
    \label{fig:ypwr_curve}
\end{figure}

\begin{figure}[htbp]
    \centering
    \includegraphics[trim={0 0 0 0.5cm},clip,width=0.45\textwidth]{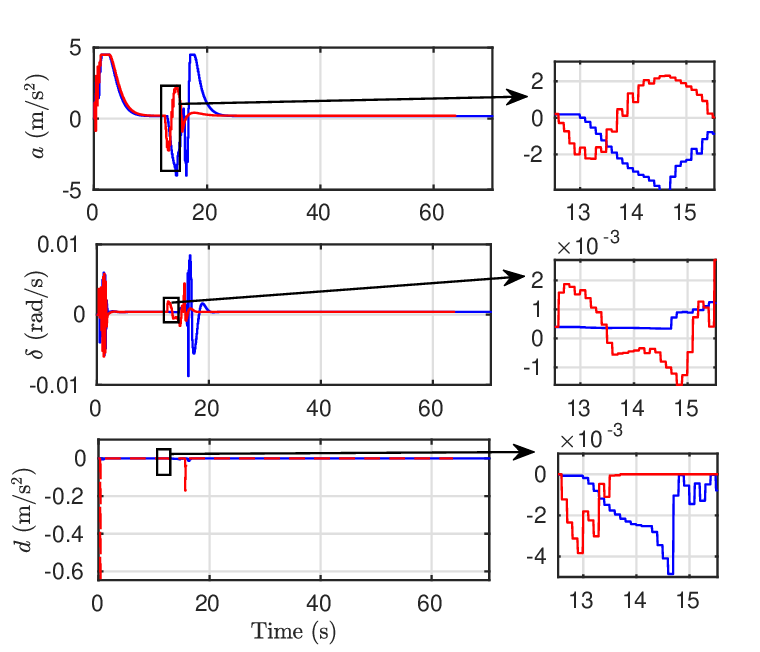}
    \caption{\footnotesize Input profile for ego EV overtake maneuver given in Figure \ref{fig:curve_ovtk}.}
    \label{fig:inp}
\end{figure}

 \begin{figure*}[t]
    \centering
    \begin{subfigure}[b]{0.475\textwidth}
\includegraphics[trim={0cm 0cm 0cm 1cm},clip,width=\linewidth]{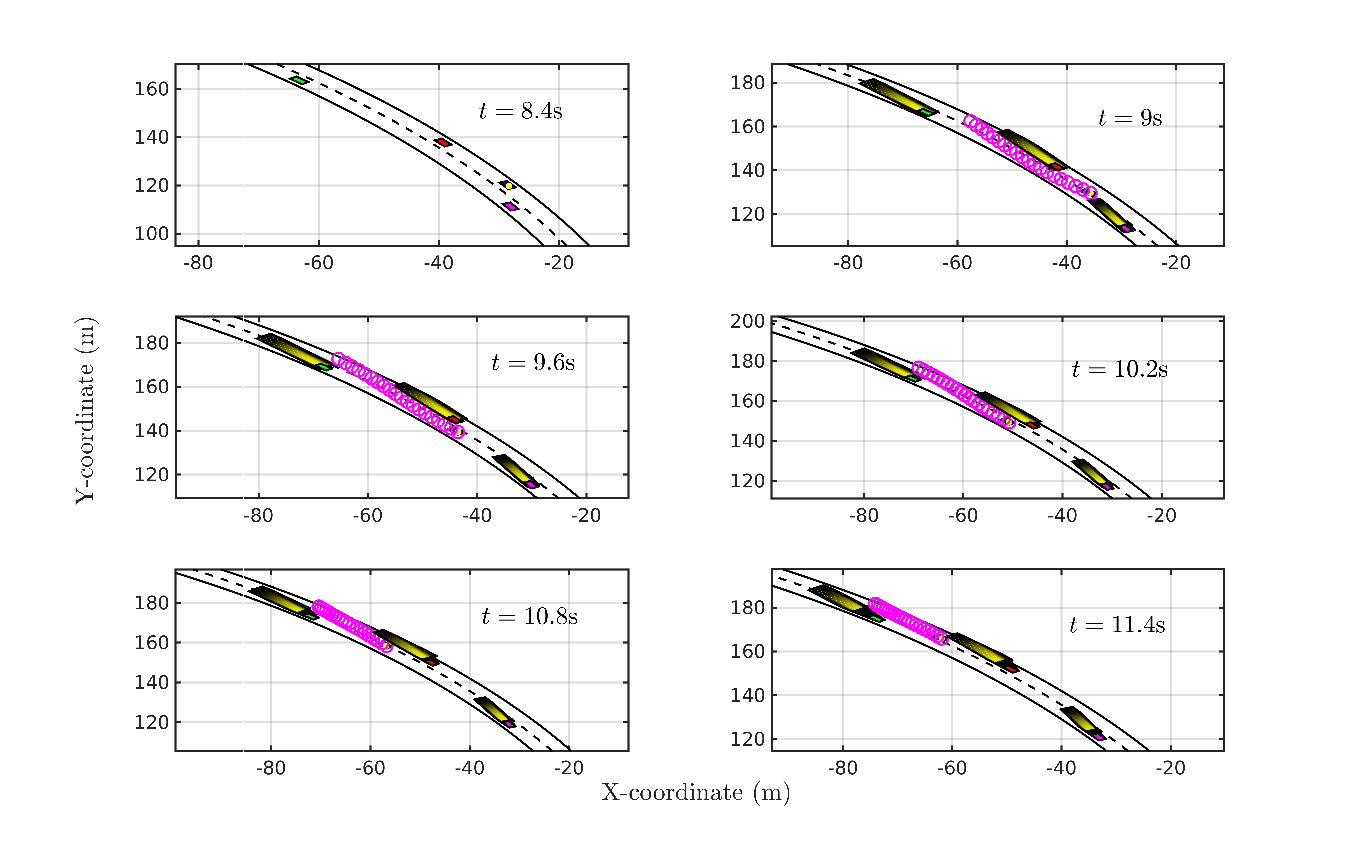}  \caption{Evolution of vehicle position in the EU scheme.} \label{mul_traja} 
    \end{subfigure}
    \hfill
    \begin{subfigure}[b]{0.475\textwidth}
      \includegraphics[trim={0cm 0cm 0cm 1cm},clip,width=\linewidth]{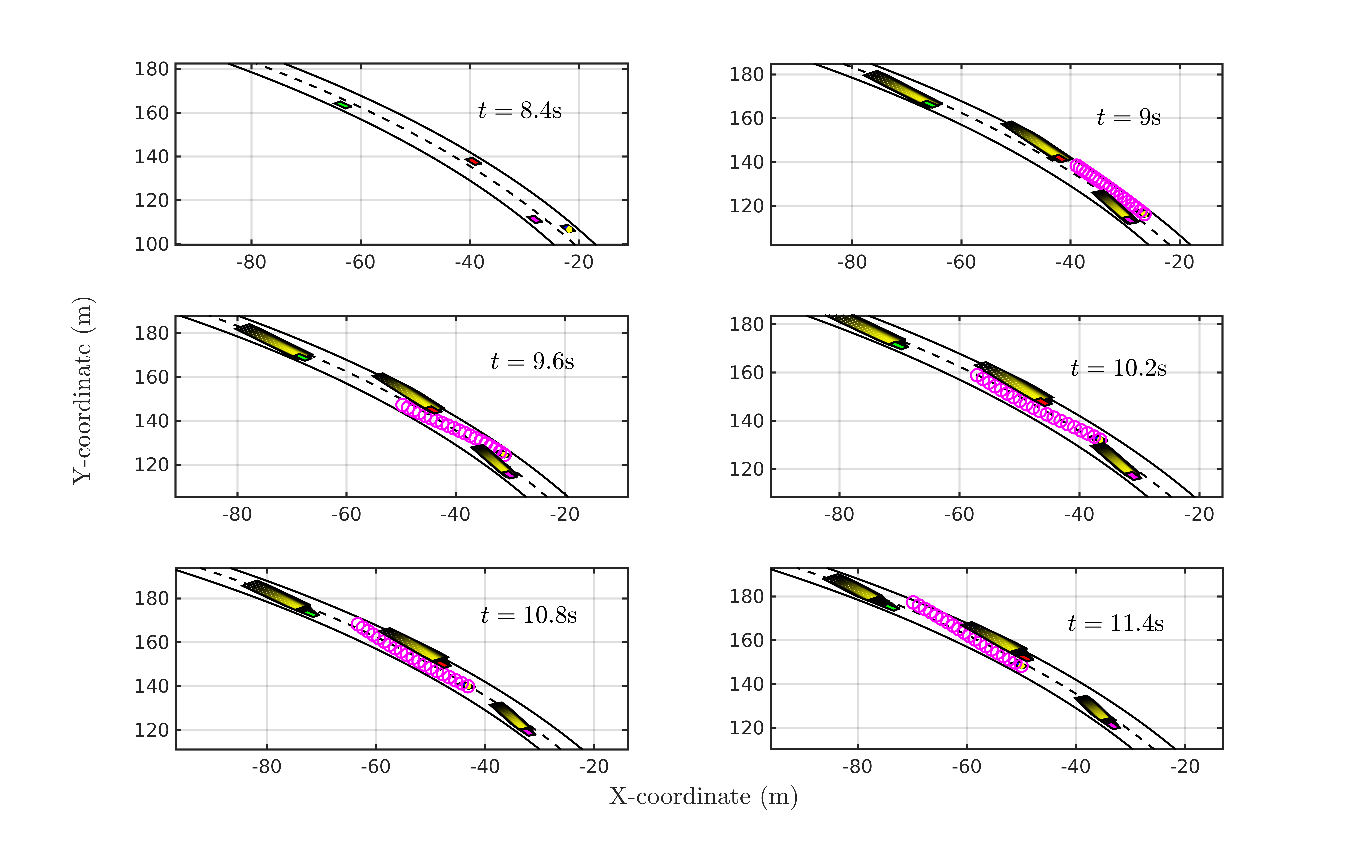}  \caption{Evolution of vehicle position in the EA scheme.} \label{mul_trajb}
    \end{subfigure}
   \caption{\footnotesize For (a) and (b), subplots show ego EV trajectory evolution without and with energy consumption consideration respectively at different time instants. Magenta circles: Trajectory of ego EV, Red, Magenta, Green polytopes: Obstacle, Yellow (increasing transparency with increase in time): Future occupancy sets.} \label{mul_traj}
\end{figure*}

Figure \ref{fig:inp} shows the input profile for both EU and EA cases. The figure shows that when the acceleration input ($a$) is positive, the deceleration input ($d$) remains negligible (or inactive). There exist different ways to deal with the simultaneous activation of forward acceleration (positive, $a$) and braking deceleration (negative, $d$). While one can add complementarity constraints to force exactly one of the inputs to be active at any given time, the inclusion of such constraints results in a mixed-integer
nonlinear program. Another option is to add a complementarity constraint of the form $a\times d\geq 0$. However, including this constraint significantly increases solution time due to the non-convexity of the product term, rendering the optimization problem intractable. For example, in the present implementation using IPOPT on a table-top PC, the solution time with the complementarity constraint is around 1s (sometimes goes up to 2s) for a single obstacle case. To mitigate this issue and maintain computational tractability, we have penalized the braking deceleration ($d$) term with a large weight in the cost function. In contrast, the computation time is around 0.15 seconds under the proposed approach, which further prioritizes regenerative braking over frictional braking and achieves desirable behavior. Note that in rare cases, the high weight on deceleration input can potentially lead to preferring unsafe maneuvers.

Table \ref{tab:energy_breakup} includes comprehensive energy consumption data needed to appreciate the various phenomena through which energy is consumed. This includes the various components of energy, namely $E_B$, $E_T$, $E_{WL}$, $E_{Long}$, $E_{Lat}$ and $E_{Brake}$ for overtaking maneuvers of EV described above. The results show that overall lateral energy dissipated in the wheel ($E_{WL}$) is substantially smaller under the proposed EA scheme compared to the EU scheme during the maneuver. This enables the ego EV to achieve a safe maneuver with a reduction in traction energy consumption in longitudinal maneuver ($E_{Long}$) relative to the EU approach of $18.3$\%  and $5.3$\%  for straight and curved road scenarios, respectively.

The total battery energy consumption ($E_B$) is shown in S.No. 12 of the Table~\ref{tab:energy_breakup}. It can be observed that during a single overtaking maneuver, $E_B$ is substantially smaller under the proposed EA scheme.  The driving range is estimated by extrapolating this data assuming one overtaking every 1 km (consistent with \cite{knoop2012quantifying,yang2018effect,asaithambi2017overtaking}) for a $54.28$ kWh battery. The estimated range under the proposed scheme shows a significant improvement compared to the EU scheme.

Note that the proposed formulation tries to minimize the overall energy consumption from the battery rather than individual components of the energy consumption. Optimizing individual components, such as the lateral energy consumption, is challenging as it is difficult to express these components as functions of states and control inputs that can be optimized. Such an analysis is beyond the scope of this paper and remains a possible direction for future research.

\begin{figure}[htbp]
    \centering
    \includegraphics[trim={0 0 0 0.5cm},clip,width=0.45\textwidth]{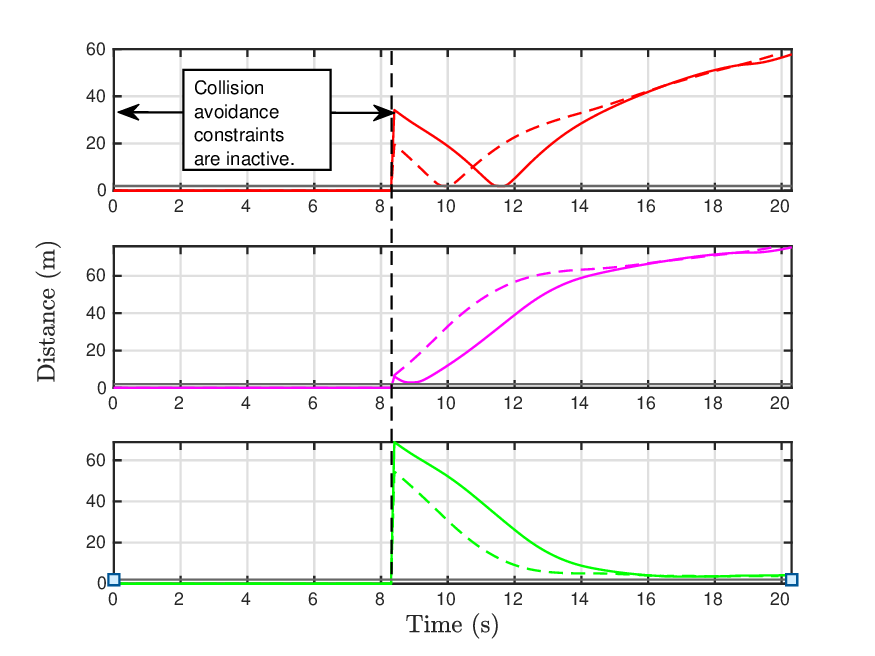}
    \caption{\footnotesize $\dd_{\EO}$ profile corresponding same colored obstacle for ego EV maneuver given in Figure \ref{mul_traj}. Dashed: EU and Solid: EA. Grey: $\dd_{\safe}$}
    \label{fig:dis_multi}
\end{figure}

\begin{figure}
    \centering
    \includegraphics[trim={0 0 0 0cm},clip,width=0.45\textwidth]{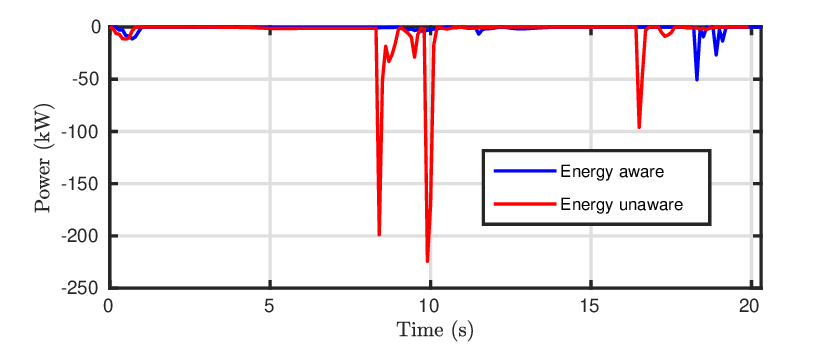}
    \caption{\footnotesize Power consumed by wheel lateral forces ($P_{WL}$) profile for ego EV overtake maneuver given in Figure \ref{mul_traj}.}
    \label{fig:ypwr_multi}
\end{figure}

\begin{figure*}[t]
    \centering
    \begin{subfigure}[b]{\textwidth}
    \includegraphics[trim={0 0cm 0 1cm},clip,width=\textwidth,height=7cm]{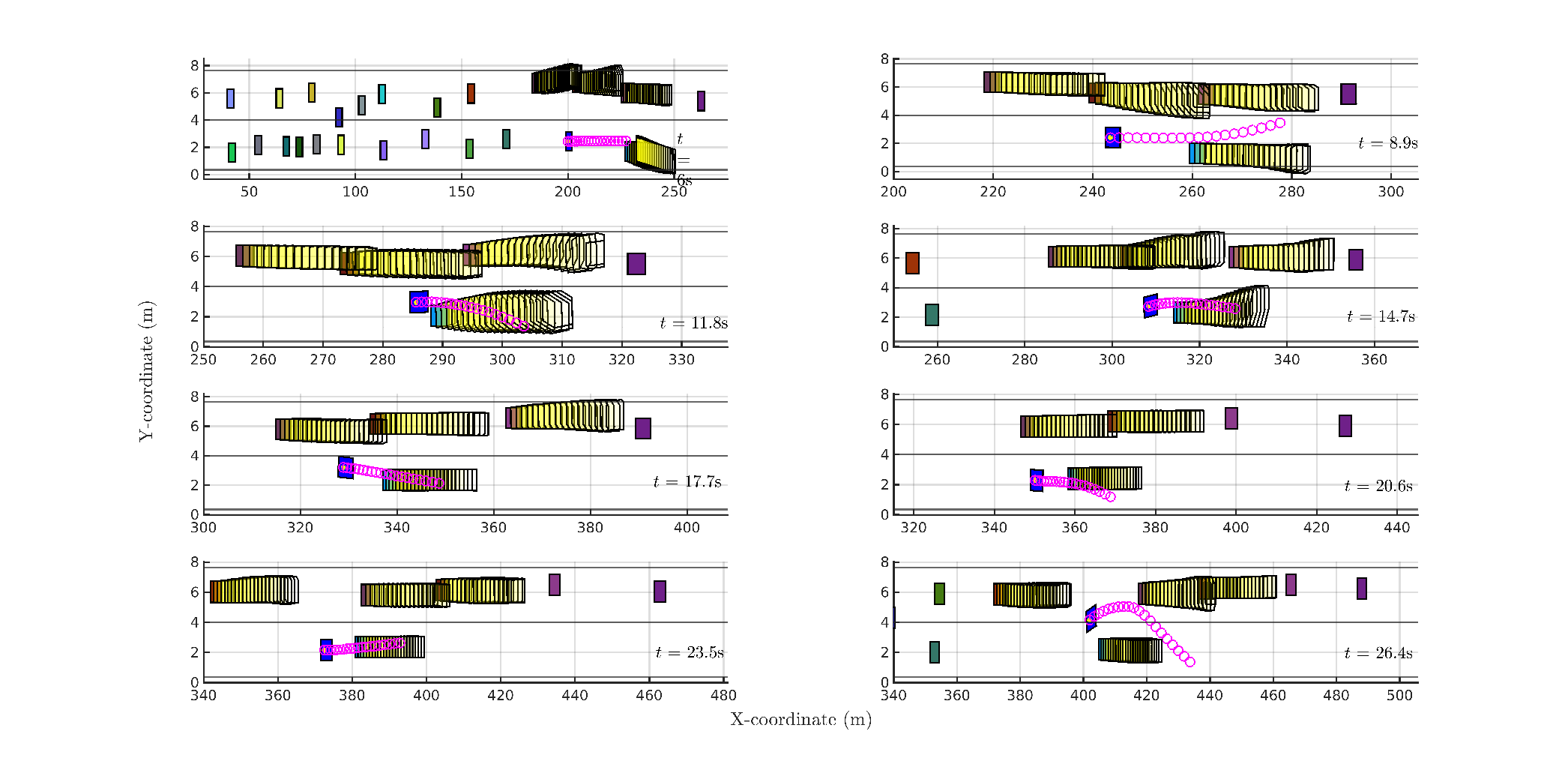}
    \end{subfigure}
    \hfill
      \begin{subfigure}[b]{\textwidth}
 \includegraphics[trim={0 0cm 0 1cm},clip,width=\textwidth,height=7cm]{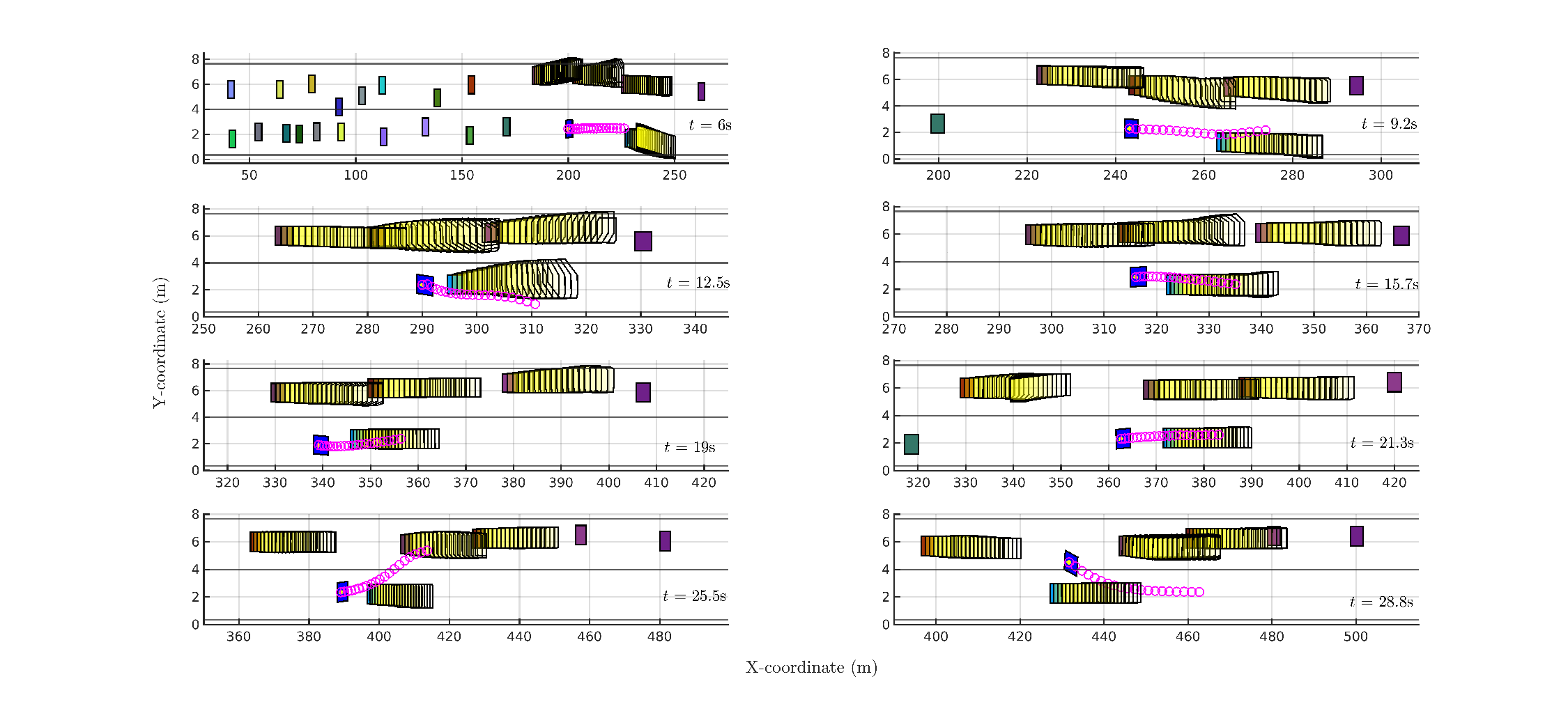}
    \end{subfigure} 
    \caption{\footnotesize Top and bottom subplots show ego EV trajectory evolution without and with energy consumption consideration, respectively, at different time instants. Blue: Ego vehicle, Magenta: Trajectory of ego EV, Yellow (increasing transparency with increase in time): Future occupancy set.}
    \label{fig:traffic_ovtk}
\end{figure*}

\begin{figure}
    \centering
    \includegraphics[trim={0 0 0 0cm},clip,width=0.45\textwidth]{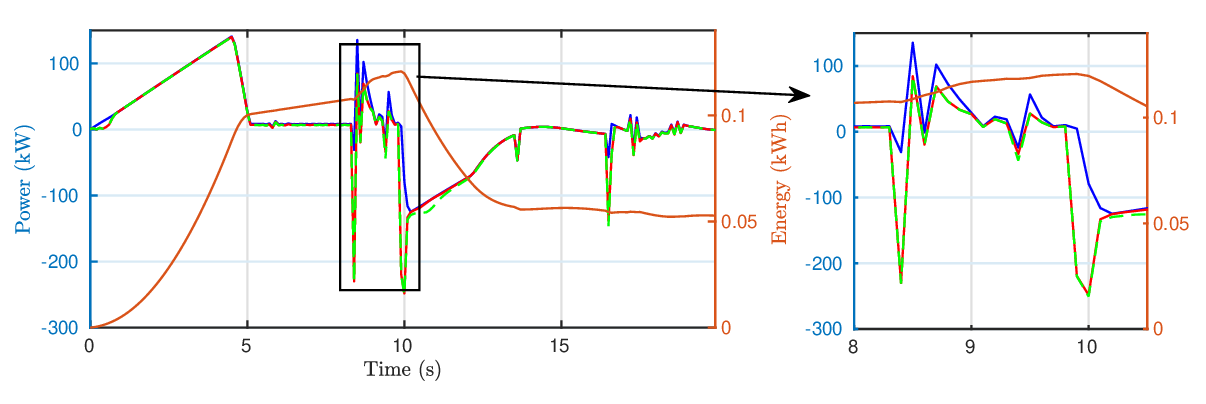}
    \includegraphics[trim={0 0 0 0cm},clip,width=0.45\textwidth]{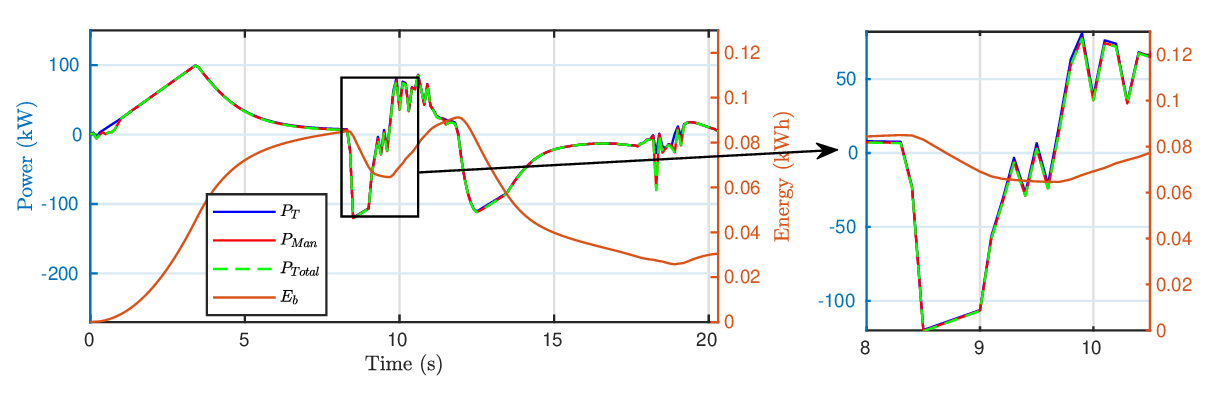}
    \caption{\footnotesize Power flow profile for (top) EU and (bottom) EA planar maneuver given in Figure \ref{mul_traj}.}
    \label{fig:pwr_ovtk_multi}
\end{figure}

\subsubsection{Multiple Obstacles} 
In this scenario, the ego EV is expected to maneuver and reach target states in the presence of multiple obstacles on a curved road with radius of curvature 250m. This scenario shows the efficacy of the proposed approach for short-range motion planning. The target state, $x_{\dest}$, for the ego EV is set to $[250,2,0,0,0,0,0.9,0,0,0]^\intercal$ and $N_s$ is set to $85$ to ensure $90\%$ confidence ($\beta=0.1$) for actual change in obstacle position to lie in $\Omega_t^*(W_t^*)$ with the probability of $0.9$ ($\epsilon=0.1$). Additionally, weights on jerk costs are omitted for this scenario. Figure \ref{mul_traj} shows the maneuver of the ego EV for EU (left) and EA (right) schemes, respectively. The green and magenta obstacles are behind and ahead of the ego EV respectively. Both obstacles move nearby along the center of the adjacent lane, and the red obstacle is moving ahead of the ego EV in the same lane as the ego EV. 

Despite this challenging environment, the ego EV avoids collision with every obstacle ($\dd_{\safe}$ chosen to be 2m) as shown in Figure \ref{fig:dis_multi}, and completes the maneuver in $19.9$s and $20.3$s for the EU and the EA cases, respectively. There exists a significant amount of $P_{WL}$ during $8.5$-$10$s as the ego EV overtakes the red obstacle and during $16$-$17$s as the ego EV maneuvers adjacent to the green obstacle for the EU case. This dissipation is large in amount for the EA case compared to the EU case as shown in Figure~\ref{fig:ypwr_multi}. The corresponding power flow and battery energy consumption ($E_B$) profile for both schemes is given in Figure \ref{fig:pwr_ovtk_multi}. Battery energy consumed under the EA scheme is $0.0339$ kWh, while energy consumed is $0.0578$ kWh for the EU scheme for this example. 

\subsubsection{Real-world traffic} 
Here, we validate the behavioral performance of the proposed framework in realistic road scenarios extracted from real traffic data. For this purpose, an open traffic dataset, namely NGSIM US-101 Vehicle Trajectories and Supporting Data, is used. Scenario selection is conducted to ensure that each scenario includes at least one lane change. This is followed by the extraction of trajectory data for all vehicles that correspond to the lanes over which lane changes were observed. The proposed energy-aware safe motion planning approach is then applied to these scenarios. The result for one such scenario is discussed next. The target or destination state, $x_{\dest}$, for the ego EV is set to $[0, -1.5, 0, 20, 0, 0, 0.9, 0, 0, 0]^\intercal$ in this case. The scenario is simulated until the effect of surrounding vehicles on the behavior of the ego EV is captured.  The value of $N_s$
is chosen to be 85 to ensure $90\%$ confidence ($\beta = 0.1$) for actual change in obstacle position to lie in $\Omega_t^*(W_t^*)$ with the probability of $0.9$ ($\epsilon=0.1$).

Figure \ref{fig:traffic_ovtk} shows the ego EV trajectories for both the EA and EU cases, respectively. Here, ego EV is simulated to travel a distance of 334 m for the former case and 305 m for the latter case to capture the effect of surrounding vehicles. The corresponding power consumption by wheel lateral forces ($P_{WL}$ ) profile for both schemes is given in
Figure \ref{fig:ypwr_ngsim}. It can be observed that the ego EV carries out safe planar maneuvers in these cases, and the energy dissipation due to wheel lateral forces is lower in the EA case. The ego EV covers the traveled distance in 29.3
seconds with energy consumption of 0.1422 kWh for the EU and in 32 seconds with 0.1042 kWh for the EA scheme. Thus, the proposed scheme can achieve a significant reduction in energy consumption with an increase
in travel time of 2.7 seconds. 

\begin{figure}
    \centering
    \includegraphics[trim={0 0 0 0cm},clip,width=0.45\textwidth]{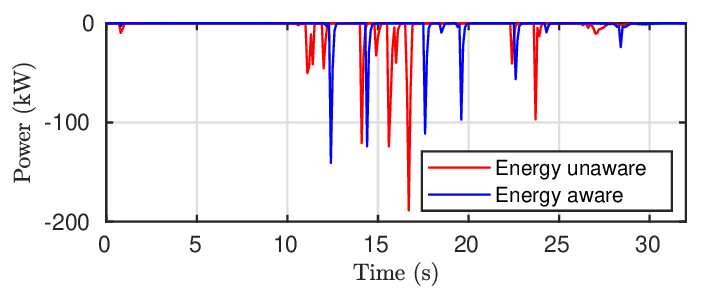}
    \caption{\footnotesize Power consumed by wheel lateral forces ($P_{WL}$) profile for ego EV overtake maneuver given in Figure \ref{fig:traffic_ovtk}.}
    \label{fig:ypwr_ngsim}
\end{figure}

The mean computation time for solving the optimization problem \eqref{eq:MPC_main} together with confidence intervals for different numbers of obstacles and prediction horizons are shown in Figure~\ref{sol_time}. The distribution for the case with $N=20$ and with three obstacles is denoted by $3b$, and so on. Although the size of the optimization problem scales with the increase in the number of obstacles, the resulting increase in computation time is not significant. 

We emphasize that our results are meant to demonstrate the feasibility and functionality of our approach rather than its real-time implementation. The computation time reported here is just to illustrate the relative impact of the prediction horizon and the number of obstacles. Our optimization program is written in Python, which is a high-level programming language, with basic solver settings. A customized C++ or Python implementation together with multi-core programming to parallelize the tasks is likely to improve computation time by an order of magnitude. We are working towards a realistic implementation of the proposed approach as a follow-up work.

\begin{figure}[htbp]
\centering
\includegraphics[trim={0 0cm 0 0cm},clip,width=0.4\textwidth,keepaspectratio]{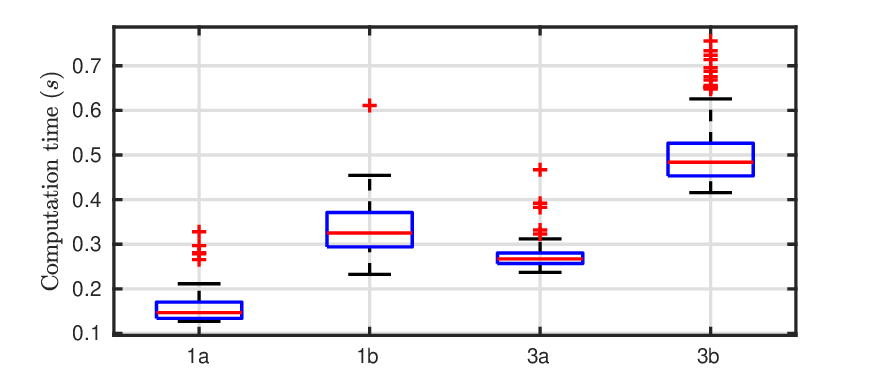}
\caption{\footnotesize Distribution of time required to solve the optimization problem \eqref{eq:MPC_main}. On the $x$-axis, the number $1$ and $3$ indicate number of obstacles, $a$ indicates prediction horizon $N=10$ and $b$ indicates $N=20$.}\label{sol_time}
\end{figure}  

\subsection{Analysis of Parameters}

\subsubsection{Effect of slack variable ($\xi$)} The variation of the slack variable throughout the trajectory duration for the ego EV, as depicted in Figure \ref{fig:strght_ovtk} and \ref{fig:curve_ovtk}, is illustrated in Figure \ref{fig:slack_profile}. Furthermore,  Figure \ref{fig:slack_profile2} presents the slack variable values across numerous simulations of overtaking scenarios on a straight road. It is noticeable that the slack variable remains nearly zero most of the time, reaching a significant magnitude in only 1.8\% of the optimization cases. Despite the non-zero slack variable, safety is maintained across all simulations due to the robust occupancy-based formulation. We have not included a solution to these rare violations. However, in practical implementation, an obvious way to maintain safety is to find a safety boundary in position-velocity space. In these rare cases, the EV then switches to a sub-optimal operating mode whenever it finds itself in an unsafe region in the position-velocity space.

\begin{figure}[htbp]
    \centering
    \includegraphics[trim={0 0 0 0cm},clip,width=0.45\textwidth]{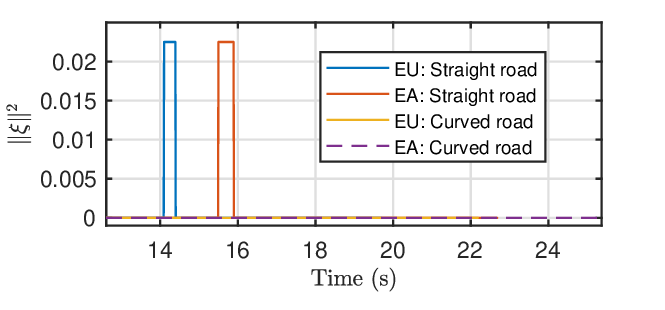}
    \caption{\footnotesize Slack variable ($\xi$) profile for the ego EV trajectory given in Figure \ref{fig:strght_ovtk} and \ref{fig:curve_ovtk}.}
    \label{fig:slack_profile}
\end{figure}

\begin{figure}[htbp]
    \centering
    \includegraphics[trim={0 0 0 0.25cm},clip,width=0.45\textwidth]{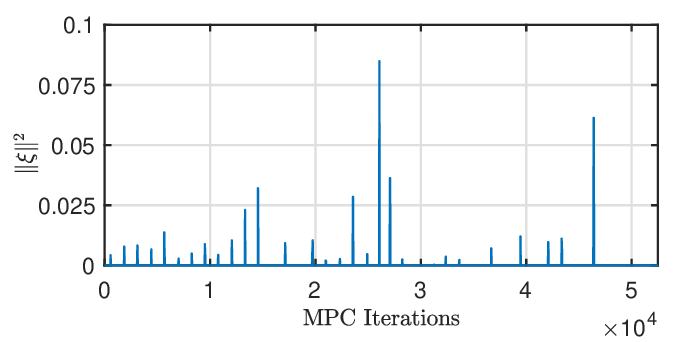}
    \caption{\footnotesize Slack variation across MPC runs for various overtake scenarios on a straight road.}
    \label{fig:slack_profile2}
\end{figure}

\subsubsection{Effect of $\epsilon$ and $\beta$} To analyze the effect of these parameters, we conducted 50 Monte Carlo simulation runs for overtake scenarios on curved road with two different sets of parameters. The corresponding relevant metrics are present in Table \ref{tab:metrics}. The metric “conservative” indicates the conservative behavior of the ego EV, representing scenarios where the ego EV cannot find a feasible space to plan a trajectory and thus applies the brakes to maintain a safe distance. This occurs when the predicted occupancy of an obstacle covers the majority of the road space ahead of the ego EV.

\begin{table}[htbp]
    \centering
    \begin{tabular}{|p{2cm}|p{2cm}|p{2.3cm}|}
    \hline
   & $\epsilon=0.1, \beta =0.1$&$\epsilon=0.01, \beta=0.01$ \\ \hline
         Failures& 3& 0\\ \hline 
       Conservative  & 0&44 \\ \hline
    \end{tabular}
    \caption{\footnotesize Metrics for energy-aware safe motion planning approach across 50 Monte Carlo simulations.}
    \label{tab:metrics}
\end{table}

Our analysis shows that a larger sample size provides higher safety guarantees but results in more conservative behavior. Conversely, a smaller sample size reduces conservative behavior and results in few failures (collisions). It is important to note that the sample size derived in Theorem \ref{thrm1} offers a loose bound, and a smaller sample size can also achieve the same probabilistic guarantees in practice.

\subsubsection{Effect on battery parameters} In all the scenarios studied in the paper, the C-rate of the battery remains within acceptable limits for commercially used battery chemistry. More specifically, under all cases of maneuvers on straight roads and curved roads, the regenerative charging current and discharging current are within 1.5 C and 3 C, respectively. 

We simulate emergency braking scenarios with varying levels of aggressiveness to further analyze the effect of the proposed approach on battery parameters. It is observed that the regenerative charging current remains within a 2.75 C-rate (maximum 400A) for harsh braking from 100 km/h to a stop over 4s, which is reasonable for a typical LFP battery used in electric vehicles.  Braking on curved roads or while turning further reduces the regenerative charging current, also reducing the energy regenerated and stored in the battery, ensuring battery parameters remain within the stated limits.


\section{Discussion and Limitations}\label{limitations}

Note that our contribution is along two dimensions: minimizing energy consumption and safe trajectory planning. Therefore, in order to obtain a fair comparison, we have kept the safety constraints intact, and compared the results when the energy consumption is penalized with the case where the the energy consumption is not penalized. While one could conduct the above comparison for alternative formulations of safety constraints, it is expected that the results will remain qualitatively similar. Since we could not find any previous work that simultaneously considers energy consumption and safety aspects, we could not include any comparative results.

Regarding the energy aspect, this paper shows that optimizing the steering input profile alongside acceleration input can significantly enhance energy efficiency in EVs. However, the presented approach assumes constant drivetrain efficiency and relies on an estimated motor efficiency map. While the qualitative results are expected to remain consistent, a more precise representation of losses in the powertrain components of EVs could yield more accurate quantitative energy consumption values. For instance, integrating a split-loss approach for modeling losses in EVs \cite{koch2022implementation}, with the proposed method might provide more precise results.

Regarding the safety aspect, various approaches for future evolution prediction of traffic participants have been explored in the past literature. Single trajectory-based approaches include physics-based models for prediction as well as maneuver-based prediction \cite{houenou2013vehicle}. Physics-based models, such as constant velocity and constant turn rate and acceleration, provide accurate predictions only over very short duration. Maneuver-based prediction can capture the intention of obstacles over longer duration but does not account for uncertainty in their future trajectories. In contrast, the reachability set-based prediction method \cite{kianfar2013safety} captures uncertainty in obstacle motion but often necessitates fail-safe planning due to its inherently conservative nature.

The number of samples required to generate an uncertainty set with a probabilistic guarantee, as presented in this paper, balances the trade-off between conservatism and efficiency in safe driving scenarios. However, similar to earlier methods, this approach does not consider interaction among traffic participants and applies primarily to sparse and quasi-stationary traffic. These are limitations compared to other probabilistic set-based methods that account for interaction. For instance, the approach in \cite{althoff2009model} models traffic participants as Markov chains and utilizes stochastic reachable sets to characterize their transition probabilities. Moreover, the inputs of traffic participants in that model are influenced by their interactions. Incorporating such interactions could further improve energy efficiency of EVs while guaranteeing safety.

\section{Conclusion}\label{conclusionn}
In this paper, we presented a data-driven robust optimization framework for energy-aware motion planning of EVs in scenarios that consist of static as well as dynamic obstacles. Simulations performed on a variety of realistic scenarios illustrate that the proposed formulation achieves a significant reduction in energy consumption with a minimal change in travel time compared to an energy-unaware scheme. 

To summarize, the primary assumption underlying the proposed formulation is that there is no interaction among traffic participants and the orientation of the obstacles is assumed to remain unchanged throughout the prediction horizon. It is further assumed that battery power consumption is a function of motor torque and angular velocity, energy regeneration efficiency is constant and power loss at the inverter and losses at wheels are negligible. In addition to relaxing the above assumptions, there are several possible directions in which the present work can be strengthened. These include (i) handling model uncertainty or mismatch between the nominal and high fidelity EV models, (ii) inclusion of interaction, (iii) formulating the collision avoidance constraints in terms of chance constraints or conditional-value-at-risk, and (iv) implementing the proposed approach in real-time hardware-in-loop simulation environments. We hope that this paper stimulates further research along the above mentioned directions.

\appendices
\section{Proof of Proposition \ref{dynamic}} \label{appendix1}

\begin{proof}
For a given translation $\omega^k \in \Omega^k_t(W^k_t)$ of the obstacle, we define the distance between the EV and the obstacle as
\begin{align*}
& \dist(p_{t+k+1},\mathbb{O}_t \! \oplus \! \{\omega^{k}\}) \\ 
= & \min_y\{||y||:A(p_{t+k+1}+y-\omega^{k})\leq b\}\\
= & \max_{\lambda_k} \{ (Ap_{t+k+1}-A\omega^{k}-b)^\intercal\lambda_k:
\lambda_k \geq 0, ||A^\intercal \lambda_k||_2\leq 1 \},
\end{align*}
which follows from \cite[Proposition 1]{zhang2020optimization}. The condition for collision avoidance is given by
\begin{align}
&\dd_{\safe}-\dist(p_{t+k+1},\mathbb{O}_{t+k+1})< 0 \quad \forall \omega^{k} \in \Omega^k_t(W^k_t) \nonumber \\
\Leftrightarrow & \max_{\substack{\omega^{k} \in \Omega^k_t(W^k_t)}}[ \dd_{\safe}- \max_{\lambda_k} \{ (Ap_{t+k+1}-A\omega^{k}-b)^\intercal \lambda_k: \nonumber \\
&\qquad \lambda_k\geq 0, ||A^\intercal \lambda_k||_2 \leq 1 \}] < 0. \label{eq:lhs_thm_main}
\end{align}

The L.H.S. of the above equation can be stated as
\begin{align*}
&\max_{\substack{\omega^{k} \in \Omega^k_t(W^k_t)}}\min_{\substack{\lambda_k\geq 0, \\ ||A^\intercal \lambda_k||_2\leq 1}}\left[ \dd_{\safe}- (Ap_{t+k+1}-A\omega^{k}-b)^\intercal \lambda_k \right]\\
& \leq \min_{\substack{\lambda_k\geq 0, \\ ||A^\intercal \lambda_k||_2\leq 1}}\max_{\substack{ \omega^{k} \in \Omega^k_t(W^k_t)}}  \left[ \dd_{\safe}- (Ap_{t+k+1}-A\omega^{k}-b)^\intercal \lambda_k \right]\\
&=\min_{\substack{\lambda_k\geq 0, \\ ||A^\intercal \lambda_k||_{2}\leq 1}} \dd_{\safe}-(Ap_{t+k+1}-b)^\intercal \lambda_k
\\ & \qquad \qquad \qquad +\max_{\omega^{k}\in \Omega^k_t(W^k_t)}(\lambda_k^\intercal  A)\omega^{k},
\end{align*}
where the first inequality is the standard min-max inequality $\left(\max\;\min\; [\;\;]\right) \leq \left(\min\;\max\;[\;\;]\right)$. Since, $\omega^k\in \Omega^k_t(W^k_t)=\{\omega^k\in \mathbb{R}^2:G^k_t\omega^k\leq h^k_t\}$, strong duality in linear programming yields
\begin{align*}
\max_{\omega^k\in \Omega^k_t(W^k_t)}(\lambda^\intercal_k A)\omega^k \Longleftrightarrow \min_{\substack{\mu_k\geq 0,\\A^\intercal \lambda_k=G^{k^\intercal}_t \mu_k}}\mu_k^\intercal h^k_t.
\end{align*}  

Therefore, the L.H.S. of \eqref{eq:lhs_thm_main} is upper bounded by the optimal value of the following optimization problem:
\begin{align*}
\min_{\lambda_k,\mu_{k}}\quad &\dd_{\safe}-(Ap_{t+k+1}-b)^\intercal \lambda_k+\mu_{k}^\intercal h^k_t\\
\text{subject to} \quad &\lambda_k\geq 0, \quad ||A^\intercal \lambda_k||_{2}\leq 1,\\
&\mu_{k}\geq 0, \quad A^\intercal \lambda_k=G^{k^\intercal}_t \mu_{k}.
\end{align*}
Therefore, the following is a sufficient condition for the collision avoidance constraint to hold in a robust manner for all perturbations that belong to the support set $\Omega(W^k)$: 
\begin{align*}
\exists \lambda_k\geq 0,\mu_{k}\geq 0:&\dd_{\safe}-(Ap_{t+k+1}-b)^\intercal \lambda_k+\mu_{k}^\intercal h^k_t<0,\\
&||A^\intercal\lambda_k||_{2}\leq 1,A^\intercal\lambda_k=G^{{k}^\intercal}_t\mu_{k}.
\end{align*}
This concludes the proof.
\end{proof}

\bibliographystyle{ieeetr}
\bibliography{reference}

\end{document}